\documentclass[aoas,preprint]{imsart}

\RequirePackage[OT1]{fontenc}
\RequirePackage{amsthm,amsmath}
\RequirePackage[numbers]{natbib}
\RequirePackage[colorlinks,citecolor=blue,urlcolor=blue]{hyperref}

\usepackage{url}
\usepackage{hyperref}
\usepackage{graphicx}
\usepackage{amsfonts}

\startlocaldefs
\numberwithin{equation}{section}
\theoremstyle{plain}
\newtheorem{theorem}{Theorem}[section]
\newtheorem{corollary}{Corollary}[theorem]
\newtheorem{lemma}[theorem]{Lemma}
\newtheorem{proposition}[theorem]{Proposition}
\newtheorem{conjecture}[theorem]{Conjecture}

\endlocaldefs

\begin{document}

\begin{frontmatter}

\title{Breaking hypothesis testing for failure rates\thanksref{T1}}

\runtitle{Breaking hypothesis testing for failure rates\thanksref{T1}}
\thankstext{T1}{Taking the UMP test for Poisson processes and applying it to others.}

\begin{aug}
\author{\fnms{Rohit} \snm{Pandey}\thanksref{t1,t2}\ead[label=e1]{ropandey@microsoft.com}\ead[label=e2]{rohitpandey576@gmail.com}},
\author{\fnms{Yingnong} \snm{Dang}},
\author{\fnms{Gil} \snm{Shafriri}},
\author{\fnms{Murali} \snm{Chintalapati}}
\and
\author{\fnms{Aerin} \snm{Kim}}

\thankstext{t1}{ropandey@microsoft.com}
\thankstext{t2}{rohitpandey576@gmail.com}

\emph{\{ropandey,yidang,gilsh,muralic,ahkim\}@microsoft.com}

\affiliation{Microsoft, Redmond}

\address{3050 152nd Ave NE\\
Redmond, WA 98052\\
\printead{e1}\\
\phantom{E-mail:rohitpandey576@gmail.com\ }\printead*{e2}}
%
\end{aug}

\begin{abstract}
We describe the utility of point processes and failure rates and the most common point process for modeling failure rates, the Poisson point process. Next, we describe the uniformly most powerful test for comparing the rates of two Poisson point processes for a one-sided test (henceforth referred to as the ``rate test''). A common argument against using this test is that real world data rarely follows the Poisson point process. We thus investigate what happens when the distributional assumptions of tests like these are violated and the test still applied. We find a non-pathological example (using the rate test on a Compound Poisson distribution with Binomial compounding) where violating the distributional assumptions of the rate test make it perform better (lower error rates). We also find that if we replace the distribution of the test statistic under the null hypothesis with any other arbitrary distribution, the performance of the test (described in terms of the false negative rate to false positive rate trade-off) remains exactly the same. Next, we compare the performance of the rate test to a version of the Wald test customized to the Negative Binomial point process and find it to perform very similarly while being much more general and versatile. Finally, we discuss the applications to Microsoft Azure. The code for all experiments performed is open source and linked in the introduction.
\end{abstract}

%

\end{frontmatter}

\newpage
\tableofcontents{}

\newpage
\section*{Introduction}
Stochastic point processes are useful tools used to model point in time events (like earthquakes, supernova explosions, machine or organ failure, etc.). Hence, they are ubiquitous across industries as varied as cloud computing, health care, climatology, etc. Two of the core properties of point processes are the rates of event arrival (how many events per unit time) and the inter-arrival time between successive events (for example, how long is a machine expected to run before it fails).

At Microsoft Azure, we have realized that machine failures are most conveniently described by point processes and have framed our KPIs (Key Performance Indicators, numbers that serve as a common language across the organization to gauge performance) around failure rates for them. Hence, we dedicate section-I to event rates for point processes. The simplest point process for modeling these random variables, and the only one that has a constant failure rate is the Poisson point process. Hence, that process will act as our base. 

Now, it is very important for us at Azure to be able to perform statistical inference on these rates (given our core KPI is devised around them) using for example, hypothesis testing. When a new feature is deployed, we want to be able to say if the failure rate is significantly worse in the treatment group that received it vis-a-vis a control group that didn't. Another field where hypothesis testing on failure rates is an active area of research is medicine (see for example, \cite{zhu}). Hence, we describe the ``uniformly most powerful test'' for comparing failure rates in section II and study its properties. In doing so, we reach some very interesting conclusions. 

In hypothesis testing, we always assume some distributions for the two groups we want to compare. A common theme across the body of research on hypothesis testing appears to be a resistance to violating this expectation too much (for example, the authors in \cite{zhu} refer to the false positive rate getting inflated when the distributional assumptions are invalidated and recommend not using the test in those scenarios). However, as we know, all models are wrong. This applies to any distributional assumption we pick to model our data - we can bet on the real data diverging from these assumptions to varying degrees. 

We therefore put our hypothesis test to the test by conducting some experiments where we willfully violate the distributional assumptions of our test (use a negative binomial point process instead of Poisson for example even though the test is devised with a Poisson assumption in mind) and study the consequences. 
We find some scenarios (non pathological) where it turns out that violating the underlying distributional assumptions of the test to a larger extent actually makes it better (where ``better'' is defined as having a better false negative to false positive rate trade off). This is covered in section III-B. Hence, we challenge this notion that violating the distributional assumptions of the test is necessarily a bad thing to be avoided.

We also reach an interesting conclusion that if we swap out the distribution of the null hypothesis with any other distribution under the sun, the trade off between the false negative rate and false positive rate remains unchanged. This conclusion holds not just for the rate test, but any one sided test. For example, if we take the famous two sample t-test for comparing means and replace the t-distribution with (for example) some weird multi-modal distribution, the false negative to false positive rate trade off will remain unchanged. These experiments are covered in section III.

Next, we measure the performance of our test, designed for the Poisson point process on a negative binomial point process and compare it to the state of the art hypothesis test designed for negative binomial point processes and find it fairs quite well. These comparisons are covered in section IV. Finally, in section V we cover the applications to Microsoft Azure and business impact of this work. All the code is open sourced \href{https://github.com/ryu577/stochproc}{and available on Github}. For example, see \href{https://github.com/ryu577/stochproc/blob/master/plots/hypoth_tst_failure_rates/plots.py}{here} for all plots you'll find in this paper and \href{https://github.com/ryu577/stochproc/blob/master/tests/hypothesis_testing_failure_rate.py}{here} for relevant tests on the library.

\section{Failure rates and the Poisson process}
Over the years, the core KPI used to track availability within Azure has shifted and evolved. For a long time, it was the total duration of customer VM (Virtual machine - the unit leased to Azure customers) downtime across the fleet. However, there were two issues with using this as a KPI:

\begin{itemize}
\item{It wasn't normalized, meaning that if we compare it across two groups with the first one having more activity, we can obviously expect more downtime duration as well.}
\item{It wasn't always aligned with customer experience. For example, a process causing many reboots each with a short duration wouldn't move the overall downtime duration by much and hence not get prioritized for fixing. However, it would still degrade customer experience especially when their workloads were sensitive to any interruptions. Customers running gaming workloads for example tend to fall into this category.}
\item{It is much harder for our telemetry to accurately capture how long a VM was down for as opposed to simply stating that there was an interruption in service around some time frame.}
\end{itemize}

The logical thing to do would be to define the KPI in terms of interruptions and that would at least address the second and third problems. However, the issue remained that it wasn't normalized. For example, as the size of the Azure fleet grows over time, we expect the number of interruptions across the fleet to increase as well. But then, if we see the number of fleet-wide interruptions increasing over time, how do we tell how much of it is due to the size increasing and how much can be attributed to the platform potentially regressing?

To address these problems, a new KPI called the `Annual Interruption Rate' or AIR was devised, which is basically a normalized form of interruptions. Before describing it, let's take a highly relevant detour into the concept of ``hazard rate". It can be interpreted as the instantaneous rate at which events from a point process are occurring, much like velocity is the instantaneous rate at which something is covering distance.

This rate can be expressed in terms of properties of the distribution representing times elapsing between the events of interest which in this case might be VM reboots. Since this time between two successive reboots is a random variable, we will express it as an upper-case, $T$. Since this notion of rates applies to any events, that is how we will refer to these `reboots'. If we denote the probability density function (PDF) of this random variable, $T$ by $f_T$ and the survival function (probability that the random variable, $T$ will exceed some value, $t$) by $S_T(t) = P(T>t)$, then the hazard rate is given by:

\begin{equation}\label{haz_rate_def}h_T(t) = \frac{f_T(t)}{S_T(t)} \end{equation}

The way to interpret this quantity is that at any time $t$, the expected number of events the process will generate in the next small interval, $\delta t$ will be given by: $h_T(t) \delta t$. You can find a derivation of this expression in appendix A. Note again that this is an instantaneous rate, meaning it is a function of time. When we talk about the Azure KPI, we're not looking to estimate a function of time. Instead, given some interval of time (like the last week) and some collection of VMs, we want to get a single number encapsulating the overall experience. In reality, the rate will indeed probably vary from instant to instant within our time interval of interest. So, we want one estimate to represent this entire profile. 

It is helpful again to draw from our analogy with velocity. If a car were moving on a straight road with a velocity that is a function of time and we wanted to find a single number to represent its average speed, what would we do? We would take the total distance traveled and divide by the total time taken for the trip. Similarly, the average rate (let's denote it by $\lambda$) over a period of time will become the number of events we are modeling divided by the total observation time interval (say $t$). 

\begin{equation}\lambda = \frac{n}{t}\label{rate_def}\end{equation}

Just as it is possible to drive a car with a steady, constant velocity, making the average and instantaneous rates the same, it is also possible to have a process where the instantaneous rate is always a constant, $\lambda$ and this is what the average rate as well will become. This special point process is called the Poisson point process (the only process with this constant rate property - henceforth denoted by $PP(\lambda)$). Chapter 5 of \cite{ross} covers this extensively. As soon as we say ``give me a single rate defining the interruptions per unit time for this data'', we're essentially asking to fit the data as closely as possible to a Poisson point process and get the $\lambda$ parameter for that process.
In section 5.3.2 of \cite{ross}, Ross mentions that reason for the name of the Poisson point process. Namely, that the number of events in any interval, $t$, $N(t)$ is distributed according to a Poisson distribution with mean $\lambda t$. The probability mass function (PMF) is defined there:

\begin{equation}\label{poisson_pmf}
P(N(t)=n) = \frac{e^{-\lambda t}(\lambda t)^n}{n!}
\end{equation}

Also, the inter-arrival times of events, $T$ follows an exponential distribution (density function $f_T(t)=\lambda e^{-\lambda t}$). This makes sense since it is the \textit{only} distribution that has a constant hazard rate with time (which is its parameter, $\lambda$). This is called the `memory-less' property (the process maintains no memory - the rate remains the same regardless of what data we observed from the distribution). We now show that equation \ref{rate_def} is consistent with the Poisson process.

\begin{proposition}
If we see $n_1, n_2 \dots n_k$ point events in observation periods, $t_1, t_2, \dots t_k$ from some data, the value of the rate parameter, $\lambda$ of the Poisson process that maximizes the likelihood of seeing this data is given by:
$$\lambda = \frac{n}{t}$$
Where $n=\sum\limits_{i=1}^k n_i$ is the total interruptions observed and $t=\sum\limits_{i=1}^k t_i$ is the total time period of observation.
\end{proposition}
\begin{proof}
Per equation \ref{poisson_pmf}, the likelihood of seeing the $i$th observation becomes:

$$L_i(\lambda) =  \frac{e^{-\lambda t_i}(\lambda t_i)^n_i}{n_i !}$$

Which makes the likelihood across all the data points:

$$L(\lambda) =  \prod\limits_{i=1}^k\frac{e^{-\lambda t_i}(\lambda t_i)^{n_i}}{n_i !}$$
Taking logarithm on both sides we get the log-likelihood function,
$$ll(\lambda) =  \sum\limits_{i=1}^k -\lambda t_i + n_i \log(\lambda t_i) -\log(n_i !)$$

To find the $\lambda$ that maximizes this likelihood, we take derivative with respect to it and set to $0$.

$$\frac{\partial ll(\lambda)}{\partial \lambda} = \sum\limits_{i=1}^k -t_i +n_i \sum\limits_{i=1}^k \frac{t_i}{\lambda t_i}=0$$
Solving for $\lambda$ we get equation \ref{rate_def} as expected when $n$ is defined as the total events and $t$ is defined as the total observation period.
\end{proof}

We can also use the fact that the inter-arrival times, $T$ are exponential to reach the same conclusion and this alternate derivation is covered in appendix B. Note that the estimator for the average rate, $\lambda$ obtained here will hold for any point process, not just $PP(\lambda)$.

\begin{proposition}
Our estimator for the rate, $\lambda$ described in equation \ref{rate_def} for a Poisson point process is unbiased and asymptotically consistent.
\end{proposition}
\begin{proof}
Let's say we observe the process for a certain amount of time, $t$. The unbiased estimator of $\lambda$ will become:

$$\hat{\lambda} = \frac{N(t)}{t}$$

The expected value of this estimator is: $E(\hat{\lambda}) = \frac{E(N(t))}{t} = \lambda$ meaning it is unbiased.

And the variance of this estimator will be:

$$V(\hat{\lambda}) = \frac{V(N(t))}{t^2} = \frac{\lambda}{t}$$ 

For a large time frame of observation, the variance in this estimator will go to $0$, making it asymptotically consistent.
\end{proof}

The `average rate' defined here is what the `AIR' (Annual Interruption Rate) KPI used within Azure is based on. It is the projected number of reboots/ other events (like blips and pauses, etc.) a customer will experience if they rent 100 VMs and run them for a year (or rent one VM and run it for 100 years; what matters is the VM-years). So, in equation \ref{rate_def}, if we measure the number of interruptions and VM-years for any scope (ex: entire Azure, a customer within Azure, a certain hardware, etc.) we get the corresponding average rate.

This definition in equation \ref{rate_def} is almost there, but is missing one subtlety related to VMs in Azure (or any cloud environment) going down for certain intervals of time as opposed to being point-events. This means that the VM might be up and running for an interval of time and then go down and stay down for some other interval before switching back to up and so on. The way to address this is to use in the denominator, the total intervals of machine UP-time only (discounting the time the machines stay down). This way, we get a failure rate per unit time the machines are actually running, which is far more useful as a KPI. In practice, this doesn't matter too much since the total time the machines spent being down is negligible compared to the time they spend being up (else we wouldn't have a business).

\section{Hypothesis testing: closed form expressions for false positive-negative trade off}
There are many questions that can be answered within the framework of hypothesis testing (see chapter 1 of \cite{lehman}). For example, we could answer the question: are the rates from two processes ``different" in a meaningful way. This is called a two-sided test. Here, we will stay focused on answering if a treatment group (group-1) has a higher failure rate than a control group (group-0). This is called a one-sided test. This question is particularly relevant in cloud environments like Azure where new software features are getting constantly deployed and we're interested in answering if a particular deployment caused the failure rate to regress. We will reference these two groups throughout this document.

A detailed description of hypothesis testing is beyond the scope of what we're discussing here. For a comprehensive treatment, refer to \cite{lehman} and \href{https://towardsdatascience.com/hypothesis-testing-visualized-6f30b18fc78f?source=friends_link&sk=cd38bd44d242bb143cc184d6c2e6f0c1}{the blog linked here} for an intuitive, visual introduction. Instead, let's simply define some terms that will be used throughout this document (they will be re-introduced with context as the need arises in the proceeding text; this is just meant as a sort of index of terms). Some of them pertain to hypothesis testing and can be looked up in the references above or in a multitude of other sources online that cover the topic. 

\begin{description}
\item[$N_0$] The number of failure events observed in data collected from group-0, the control group. We will consider multiple distributions for this variable in the proceeding discussion.
\item[$t_0$] The total observation time for group-0, the control group.
\item[$N_1$] The number of failure events observed in data collected from group-1, the treatment group. Again, multiple distributional assumptions will be considered.
\item[$t_1$] The total observation time for group-1, the treatment group.
\item[$\lambda$] The underlying failure rate of the control group. Per equation \ref{rate_def}, the unbiased estimator for this rate is: $\frac{N_0}{t_0}$
\item[$\delta \lambda$] The effect size. If we imagine that the treatment group has a worse failure rate, this is the amount by which we assume it to be worse. It is closely related the the alternate hypothesis, $H_a$ defined below.
\item[$X$] The test statistic. We take the data from the two groups and convert it to a single number. We can then observe this number from our collected data and if it's high (or low) enough, conclude a regression was caused. For example, it could be the difference in estimated rates.
\item[$H_0$] The null hypothesis of the test. We always start with the assumption of innocence and this represents the hypothesis that the treatment group does not have a worse failure rate than the control group. Further, the distributional assumptions on $N_0$ and $N_1$ made by the test are satisfied. For this paper, this will mostly mean that $N_0$ and $N_1$ are both Poisson processes, $PP(\lambda)$.
\item[$H_a$] The alternate hypothesis. In this hypothesis, we assume that the treatment group indeed has a worse failure rate than the control group. To make it concrete, we assume it's worse by the effect size, $\delta \lambda$. Like $H_0$, the distributional assumptions made by the test are assumed satisfied. This will mean for the most part that the control group follows $PP(\lambda)$ and the treatment group follows $PP(\lambda+\delta \lambda)$.
\item[$H_0'$] This is a new hypothesis we're defining. It is like $H_0$, apart from allowing the distributional assumptions on $N_0$ and $N_1$ to be different from the test. The failure rates for the two processes are still assumed to be the same for the two groups. It allows us to address the question of what happens when we use a test designed on one set of assumptions on real data that diverges from those assumptions.
\item[$H_a'$] Like $H_a$, apart from allowing the distributional assumptions on $N_0$ and $N_1$ to be different from the test. The failure rates for the two processes are still assumed to differ by $\delta \lambda$ just as with $H_a$.
\item[$X_0$] The distribution of our test statistic, $X$ under $H_0$.
\item[$X_a$] The distribution of our test statistic, $X$ under $H_a$.
\item[$Y_0$] The distribution of our test statistic, $X$ under $H_0'$.
\item[$Y_a$] The distribution of our test statistic, $X$ under $H_a'$.
\item[$\phi$] The p-value of the hypothesis test, representing the likelihood that something as or more extreme (with ``extreme" defined in the direction of $H_a$, which here means towards greater treatment failure rates) as the observed test statistic could be seen under the assumptions of $H_0$.
\item[$\hat{\alpha}$] The type-1 error rate of the test. It is the only parameter defined arbitrarily by us. Under the assumptions of $H_0$, what is the probability the test will reject it? It is the theoretical false positive rate from the test. The binary decision saying weather or not there is a regression in the treatment group is made using the indicator variable: $I(\phi < \hat{\alpha})$.
\item[$\alpha(\hat{\alpha})$] The false positive rate (FPR) for real world data when there is no difference in rates between the groups (so, under $H_0'$) and we still use $I(\phi < \hat{\alpha})$ to reject the null hypothesis. We will see in proposition \ref{prop:hat_equal} that if $H_0' \sim H_0$ then $\alpha(\hat{\alpha}) = \hat{\alpha}$.
\item[$\beta(\hat{\alpha})$] The false negative rate of our test (as a function of the type-1 error rate $\hat{\alpha}$ we arbitrarily set), defined as the probability that we will fail to reject the null hypothesis under $H_a$ or $H_a'$.
\item[$\beta(\alpha)$] The false negative rate at the value of $\hat{\alpha}$ where we get a false positive rate of $\alpha$.
\end{description}

Let's also define henceforth for a distribution $X$ (typically the test statistic in our context), $F_X(x)=P(X<x)$, the cumulative distribution function (CDF) of $X$ and $S_X(x)=P(X>x)$, the survival function of $X$.

Armed with the notation defined above, we can now describe how our hypothesis test (one sided with alternate hypothesis being that the treatment group has a higher rate) proceeds (refer to figure \ref{fig:tradeoff}):

\begin{figure}
  \includegraphics[width=0.8\linewidth]{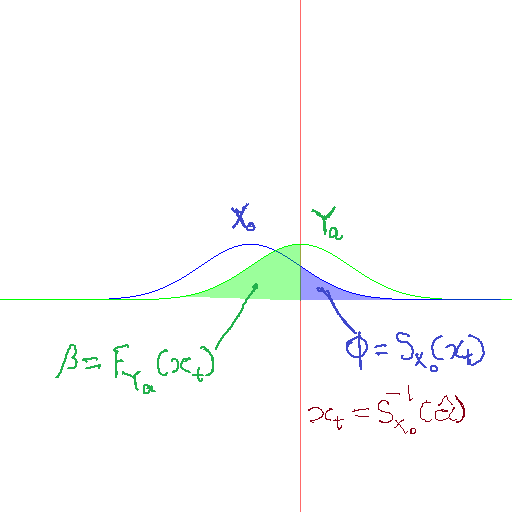}
  \caption{The false positive-false negative rate trade-off. As we increase $\hat{\alpha}$, $x_t$ which is the threshold the test statistic, $X$ needs to cross for rejecting the null increases. As a result, the p-value, $\phi$ which is the blue area reduces while the false negative rate, $\beta$ which is the green area increases.}
  \label{fig:tradeoff}
\end{figure}

\begin{description}
\item[Step 1:] Obtain the distribution, $X_0$ of the test statistic $X$ under the null hypothesis, $H_0$. This distribution is represented by the blue distribution in figure \ref{fig:tradeoff}.
\item[Step 2:] Observe the estimated value of the test statistic, $X=x$ in the data we collect. This value is represented by the red line in figure \ref{fig:tradeoff}. We assume that this test statistic is higher when the difference in rates between the treatment and control groups is higher.
\item[Step 3:] Find the probability of seeing something as or more extreme than $X=x$ under the assumptions of $X_0$. This is called the p-value, $\phi=P(X_0>x)$ and is represented by the blue area to the right of the red line in figure \ref{fig:tradeoff}.
\item[Step 4:] For some arbitrarily defined type-1 error rate (a common value is 5\%), $\hat{\alpha}$, reject the null and conclude there is a regression if $\phi<\hat{\alpha}$.
\end{description}

\begin{proposition}\label{prop:hat_equal}
Under the assumptions of the null hypothesis (the hypothesis whose distributional assumption for the test statistic is used to calculate the p-value, $\phi$), the type-1 error rate of our test ($\hat{\alpha}$) is the same as the false positive rate ($\alpha$).
\end{proposition}
\begin{proof}
The p-value will be given by:
\begin{equation}\phi = P(X_0>X) = S_{X_0}(X)\label{p_value_definition}\end{equation}

where $S_{X_0}(x)=P(X_0>x)$ is the survival function of the distribution, $X_0$.

The false positive rate, $\alpha(\hat{\alpha})$ then becomes the probability that the p-value will be lower than the type-1 error rate, $\hat{\alpha}$.

\begin{equation}\alpha(\hat{\alpha}) = P(\phi < \hat{\alpha}) = P(S_{X_0}(X)<\hat{\alpha}\label{false_positive_rate})\end{equation}

Under the assumptions of the null hypothesis however, the test statistic $X$ is distributed as $X_0$:

$$X \overset{H_0}{\sim} X_0$$

Substituting into equation \ref{false_positive_rate} we get:

\begin{align}
\alpha(\hat{\alpha}) = P(S_{X_0}(X_0)<\hat{\alpha}) \nonumber \\
= P(X_0 > S_{X_0}^{-1}(\hat{\alpha})) \nonumber \\
= S_{X_0}S_{X_0}^{-1}(\hat{\alpha})\nonumber \\ 
= \hat{\alpha} \label{hat_alpha}
\end{align}

Where in the third step, we used the fact that $S_{X_0}(x)$ is a monotonically decreasing function.
\end{proof}

\begin{corollary}
Under the null hypothesis, the p-value ($\phi$) is uniformly distributed over $(0,1)$.
\end{corollary}
\begin{proof}
From equation \ref{false_positive_rate} and the result of proposition \ref{prop:hat_equal} we have,
$$P(\phi < \hat{\alpha}) = \hat{\alpha}$$
The only distribution that satisfies this property is the uniform distribution, $U(0,1)$.
\end{proof}

In making a binary decision on weather or not there is a regression in failure rates for the treatment group, there will be a trade off between false negative (failing to reject null when its false) and false positive (rejecting null when it's true) error rates. 

To define false negative rate, we assume there is actually a difference in the failure rates for the treatment and control groups. We assumed that the test statistic follows the distribution $Y_a$ under this hypothesis. If it so happens that the distributional assumptions on $N_0$ and $N_1$ happen to be of the same form as those used to derive $X_0$ under the hypothesis $H_0$ (apart from the failure rate corresponding to $N_1$ being higher than that for $N_0$ by $\delta \lambda$), we get $Y_a \sim X_a$ but that's incidental. Referencing figure \ref{fig:tradeoff} again, we get the following proposition:

\begin{proposition}\label{prop:fnr_def}
The false negative rate of our hypothesis test described earlier as a function of the type-1 error rate we arbitrarily set is given by: $\beta = F_{Y_a}(S_{X_0}^{-1}(\hat{\alpha}))$
\end{proposition}
\begin{proof}
Refer again to figure \ref{fig:tradeoff} where the green distribution represents $Y_a$, the hypothesis where the failure rate of the treatment group is higher than that of the control group. Our test $I(\phi<\hat{\alpha})$ translates to some threshold, $x_t$ on the observed test statistic where we reject the null if $X>x_t$. Since we have per equation \ref{p_value_definition}, $\phi = S_{X_0}(X)$, we get:

\begin{equation} x_t = S_{X_0}^{-1}(\phi) \label{test_stat_tau}\end{equation}

The false negative rate then becomes the probability of the observed test statistic being below this threshold:

$$\beta(\hat{\alpha}) = P(X<x_t) = F_{X}(x_t)$$
where $F_X(x)$ is the cumulative density function of $X$.

Substituting equation \ref{test_stat_tau} and noting $X \sim Y_a$ under current assumptions we get:

\begin{equation}
\beta(\hat{\alpha}) = F_{Y_a}S_{X_0}^{-1}(\hat{\alpha})
\label{beta_def}
\end{equation}
\end{proof}

As a special case, if in the alternate hypothesis, the distributional assumptions of $N_0$ and $N_1$ are maintained, we have $Y_a \sim X_a$ and equation \ref{beta_def} becomes:
\begin{equation}
\beta(\hat{\alpha}) = F_{X_a}S_{X_0}^{-1}(\hat{\alpha})
\label{beta_definition2}
\end{equation}

Alternately, we can also proceed as follows to prove proposition \ref{prop:fnr_def}:
\begin{proof}
The false negative rate, $\beta$ is defined as the probability of failing to reject the null hypothesis conditional on it being true. The probability of failing to reject the null is $P(\phi > \hat{\alpha})$. Using equation \ref{p_value_definition}, this becomes:

$$\beta(\hat{\alpha}) = P(S_{X_0}(X) > \hat{\alpha})$$

But, under the alternate hypothesis we have:

$$X \overset{H_a}{\sim} Y_a$$
This implies

\begin{align}
\beta(\hat{\alpha}) = P(S_{X_0}(Y_a)>\hat{\alpha}) \nonumber \\
= P(Y_a< S_{X_0}^{-1}(\hat{\alpha})) \nonumber \\
= F_{Y_a}(S_{X_0}^{-1}(\hat{\alpha})) \nonumber
\end{align}
Where in the second equation we used the fact that the survival function, $S_{X_0}$ is a decreasing function.
\end{proof}

What if we assumed some distributions for $N_0$ and $N_1$, leading to the null hypothesis, $H_0$. In real life, $N_0$ and $N_1$ follow some other distribution, while still having the same rates for the two processes. This leads to another null hypothesis, $H_0'$. The test statistic under the two hypotheses:

$$X \overset{H_0}{\sim} X_0$$
$$X \overset{H_0'}{\sim} Y_0$$
In this case, the violation of the distributional assumption causes the false positive and type-1 error rates to diverge unlike equation \ref{hat_alpha}.

\begin{proposition}
Under $H_0'$, the false positive rate as a function of the type-1 error is given by: 

\begin{equation}\alpha(\hat{\alpha}) = S_{Y_0}S_{X_0}^{-1}(\hat{\alpha})
\label{alpha_general}
\end{equation}
\end{proposition}
\begin{proof}
Left to the reader, proceed similarly to equation \ref{hat_alpha}.
\end{proof}

\begin{corollary}\label{alpha_adjustment}
If we're applying a hypothesis test that is designed under $H_0$ that involves the test statistic following a distribution given by $X_0$ where as we expect to encounter data where we know the null hypothesis is actually going to follow the distribution $Y_0$. If we're then targeting a false positive rate of $\alpha$, we should set the type-1 error rate, $\hat{\alpha}$ to:

\begin{equation}
\hat{\alpha} = S_{X_0}S_{Y_0}^{-1}(\alpha)
\end{equation} 
and the probability of observing something as or more extreme than the test statistic under the distributional assumptions of $Y_0$ becomes:

\begin{equation}
\phi' = S_{X_0}S_{Y_0}^{-1}(\phi)
\end{equation} 
where $\phi$ is the p-value under $H_0$.
\end{corollary}

\begin{corollary}
If the effect size, $\delta \lambda=0$, the false negative rate, $\beta(\alpha)$ as a function of the false positive rate, $\alpha$ is given by:
$$\beta(\alpha)=1-\alpha$$
 \end{corollary}
\begin{proof}
The result follows from equations \ref{alpha_general} and \ref{beta_def} and noting that if $\delta\lambda=0$ then $Y_a \sim Y_0$ (assuming the only difference between $H_0'$ and $H_a'$ is the effect size).

\begin{align*}
\beta(\hat{\alpha}) = F_{Y_a}S_{X_0}^{-1}(\hat{\alpha}) \\
= F_{Y_0}S_{X_0}^{-1}(\hat{\alpha})\\
 = 1-S_{Y_0}S_{X_0}^{-1}(\hat{\alpha})\\
=1-\alpha(\hat{\alpha})
\end{align*}
Note that this $\beta$ profile is equivalent to tossing a coin with $\alpha$ being the probability of heads and rejecting the null if we get heads.
\end{proof}

\begin{proposition}\label{fnr_fpr}
Under $H_a'$, the false negative rate as a function of false positive rate is given by: $\beta(\alpha) = F_{Y_a}S_{Y_0}^{-1}(\alpha)$
\end{proposition}
\begin{proof}
From equation \ref{alpha_general} we have:

$$\hat{\alpha} = S_{X_0}S_{Y_0}^{-1}(\alpha)$$

Substituting into equation \ref{beta_def} we get:

\begin{align}
\beta(\alpha) = F_{Y_a}S_{X_0}^{-1}S_{X_0}S_{Y_0}^{-1}(\alpha)\nonumber \\
= F_{Y_a}S_{Y_0}^{-1}(\alpha)\label{beta_general}
\end{align}
\end{proof}

Equations \ref{beta_def} and \ref{beta_general} define the false negative rate as a function of the type-1 error rate (which we set arbitrarily) and false positive rate (which we get from the real data) respectively. We should expect:

\begin{itemize}
\item{The higher we set the type-1 error rate, $\hat{\alpha}$, the more prone to fire our hypothesis test, $I(\phi<\hat{\alpha})$ is becoming, rejecting the null hypothesis more easily. So, the false positive rate should become higher when we do this. Hence, the type-1 error rate should  be an increasing function of the type-1 error rate we set.}
\item{By a similar argument, the higher we set the type-1 error rate, the lower our false negative rate should become, since the test is only more likely to reject the null.}
\item{From the above two arguments, it follows that the false negative rate should always be a decreasing function of the false positive rate.}
\end{itemize}

Here, we prove the second and third conclusions above.

\begin{proposition}
The false negative rate can only be a decreasing function of the type-1 error rate we set and false positive rate our test consequently provides.
\end{proposition}
\begin{proof}
We have from equation \ref{beta_general},

$$\beta(\alpha) = F_{Y_a}S_{Y_0}^{-1}(\alpha)$$
Differentiating with respect to $\alpha$ we get:

$$\frac{\partial \beta(\alpha)}{\partial \alpha} = \frac{\partial F_{Y_a}S_{Y_0}^{-1}(\alpha)}{\partial S_{Y_0}^{-1}(\alpha)}\frac{\partial S_{Y_0}^{-1}(\alpha)}{\partial \alpha}$$

$$=f_{Y_a}(S_{Y_0}^{-1}(\alpha))\frac{\partial S_{Y_0}^{-1}(\alpha)}{\partial \alpha}$$
Where $f_{Y_a}$ is the probability density function of $Y_a$ which will always be positive while the second term will be negative since the survival function of any distribution is always decreasing and so is its inverse. Hence we always have $\frac{\partial \beta(\alpha)}{\partial \alpha}$ being a monotonically decreasing function and we can similarly prove the same for 
$\frac{\partial \beta(\hat{\alpha})}{\partial \hat{\alpha}}$
\end{proof}

Since our hypothesis test is basically an oracle, that is supposed to alert us when there is a difference and not fire when there isn't, there is a good argument to the assertion that the $\beta(\alpha)$ profile is all that matters when comparing various hypothesis tests. If a test produces a better false negative rate for any given actual false positive rate ($\alpha$) than another (everything else being equal), it should be preferred. Such a test is called ``more powerful'' since the power is defined as $1-\beta(\alpha)$.

\subsection{The most powerful test for failure rates}

As mentioned in section I, talking of failure rates is synonymous with fitting a Poisson process to whatever point process we're modeling and finding the rate, $\lambda$ of that Poisson process. This gives us a good starting point for comparing failure rates since we now have not just a statistic, but an entire point process to work with.

For comparing the rate parameters of two Poisson point processes, there exists a uniformly most powerful (UMP) test (see section 4.5 of \cite{lehman}). This test is mathematically proven to produce the best false negative rate (power) given any false positive rate, effect size and amount of data (in this context, observation period). We will describe the test here, but refer to \cite{lehman} for a detailed treatment and why this is the ``Uniformly most powerful (UMP)" test for comparing Poisson rates.

To review, we have two Poisson processes. We observe $n_1$ events in time $t_1$ from the first and $n_2$ events in time $t_2$ from the second one. Hence, the estimates for the two failure rates we want to compare are: $\lambda_1=\frac{n_1}{t_1}$ and $\lambda_2 = \frac{n_2}{t_2}$. The proceeding theorem will help convert this hypothesis testing problem into a simpler one, but we need a few Lemmas before we get to it.

\begin{lemma}\label{sum_of_poissons}
If we sum two independent Poisson random variables with means $\mu_0$ and $\mu_1$, we get another Poisson random variable with mean $\mu_0+\mu_1$.
\end{lemma}
\begin{proof}
Let $N_0$ and $N_1$ denote the two Poisson random variables. Conditioning on the value of $N_0$,
$$P(N_0+N_1=n) = \sum\limits_{j=0}^{n}P(N_0+N_1=n|N_0=j)P(N_0=j)$$
$$=\sum\limits_{j=0}^{n} P(N_1=n-j | N_0=j)P(N_0=j)$$

Since $N_0$ and $N_1$ are independent by definition,
$$ = \sum\limits_{j=0}^{n} P(N_1=n-j)P(N_0=j)$$
$$ = \sum\limits_{j=0}^{n} \frac{e^{-\mu_1}\mu_1^{n-j}}{(n-j)!} \frac{e^{-\mu_0}\mu_0^{j}}{j!}$$
$$=\frac{e^{-(\mu_1+\mu_2)}}{n!} \sum\limits_{j=0}^{n} {n \choose j}\mu_0^j \mu_1^{n-j}$$
$$ = \frac{e^{-(\mu_0+\mu_1)} (\mu_0+\mu_1)^n}{n!}$$
\end{proof}

The Binomial distribution with parameters $n$ and $p$ is defined as the number of heads we get when we toss a coin with $p$ being its probability of heads $n$ times (represented henceforth by $B(n,p)$), we have the following lemma:

\begin{lemma}\label{lem:conditional_binom}
Given that two Poisson processes with rates $\lambda_0$ and $\lambda_1$ which we observe for periods $t_0$ and $t_1$; conditional on a total of $n$ events observed, the number of events, $N_1 = n_1$ from the second process is a Binomial distribution with parameters $n$ and $p=\frac{\lambda_1 t_1}{\lambda_0 t_0 + \lambda_1 t_1}$.
\end{lemma}

\begin{proof}
$N_0$ and $N_1$ represent the random numbers describing the number of events from the two processes. We have by Bayes theorem:

$$P(N_1=j | N_0+N_1=n) = \frac{P(N_1=j \cap N_1+N_0=n)}{P(N_1+N_0=n)}$$

Since the two processes are independent,
$$=\frac{P(N_1=j)P(N_0=n-j)}{P(N_1+N_0=n)}$$

The number of events, $N(t)$ in a time interval of length $t$ from a Poisson process with rate $\lambda$ is Poisson distributed with mean $\lambda t$. Also, using the result of Lemma \ref{sum_of_poissons},

$$=\frac{\left(\frac{e^{-\lambda_0 t_0} (\lambda_0 t_0)^j}{j!}\right)\left(\frac{e^{-\lambda_1 t_1} (\lambda_1 t_1)^{n-j}}{(n-j)!}\right)}{\left(\frac{e^{-(\lambda_0 t_0+\lambda_1 t_1)} (\lambda_0 t_0+\lambda_1 t_1)^n}{ n!}\right)}$$

\begin{equation} ={n \choose j} \left(\frac{\lambda_1 t_1}{\lambda_1 t_1+\lambda_0 t_0}\right)^j \left(\frac{\lambda_0 t_0}{\lambda_1 t_1+\lambda_0 t_0}\right)^{n-j}\label{eqn:binom_cond}\end{equation}

which is the Binomial probability mass function (PMF) as required.
\end{proof}

\begin{corollary}\label{null_rate_test}
If two Poisson processes have the same rate, $\lambda$ and are observed for periods $t_1$ and $t_2$, then conditional on observing $n$ events from both processes, the number of events from the first process is a Binomial distribution with parameters $n$ and $p=\frac{t_1}{t_1+t_0}$.
\end{corollary}
\begin{proof}
Substitute $\lambda_1=\lambda_0$ into equation \ref{eqn:binom_cond} above.
\end{proof}

Per Corollary \ref{null_rate_test}, we've managed to get rid of rate, $\lambda$ if the two processes are identical (which is a requirement for the null hypothesis), a nuisance parameter. This ensures our hypothesis test for failure rates will work the same regardless of the base failure rate for the two processes, $\lambda$. So, conditional on the total events from the two processes being $n=n_0+n_1$ (which is something we observe), asking if the second process has a higher failure rate becomes equivalent to asking if the conditional Binomial distribution has a higher value of the parameter, $p$ than $\frac{t_1}{t_1+t_0}$ as the null hypothesis would suggest. We have thus reduced the two sample rate test to a one sample Binomial test on the probability of success, $p$. 

\subsubsection{The one-sample Binomial test}
To get the p-value (probability of being able to reject the null hypothesis), we ask - ``what is the probability of seeing something as or more extreme than the observed data per the null hypothesis". Here, ``extreme" is defined in the direction of the alternate hypothesis. So, if we observe $n_1$ heads out of $n$ tosses in our data and our null hypothesis is that the probability of heads is $p$, then the p-value, $\phi$ becomes the probability of seeing $n_1$ or more heads if the probability of seeing heads in a single toss was $p$. So we get (where $X_0 \overset{H_0}{\sim} \text{B}(n,p)$):

\begin{equation}\phi = P(X_0 \geq n_1) = \sum\limits_{j=n_1}^{n}{n \choose j} p^{j}(1-p)^{n-j}\label{binom_tst}\end{equation}

\subsubsection{Back to comparing Poisson rates}
\begin{theorem}\label{rate_test}
Given two Poisson processes with rates $\lambda_0$ and $\lambda_1$, under the null hypothesis - $H_0: \lambda_0=\lambda_1$, conditional on observing a total of $n$ events and alternate hypothesis, $H_a: \lambda_1 > \lambda_0$ with a similar condition, if we observe $n_0$ events from the first process in time $t_0$ and $n_1=n-n_0$ events from the second in time $t_1$, the p-value, $\phi$ is given by:

\begin{equation}\phi = \sum\limits_{j=n_1}^{n_1+n_0}{n_1+n_0 \choose n_1} \left(\frac{t_1}{t_1+t_0}\right)^{j}\left(\frac{t_0}{t_1+t_0}\right)^{n_1+n_0-j}\label{poisson_tst}\end{equation}

We can then pick a type-1 error rate, $\hat{\alpha}$ and reject the null if $\phi<\hat{\alpha}$.
\end{theorem}
\begin{proof}
Per corollary \ref{null_rate_test}, conditional on observing a total of $n$ events from both processes and failure rates being the same, the distribution of events from the second process, $n_0$ is $\text{B}(n,p=\frac{t_1}{t_1+t_0})$. Substituting this into equation \ref{binom_tst}, the result follows.
\end{proof}

Note that for our simple, one-sided test, the Poisson rate test can be readily swapped with the Binomial test. However, there is some subtlety when dealing with two-sided tests and confidence intervals. This is covered in \cite{r_paper}.

\subsection{False positive negative trade off}
Now that we have described our test for comparing failure rates, we will evaluate the false positive to false negative rate trade off function ($\beta(\alpha)$) under the assumptions of the null hypothesis (the hypothesis under whose test statistic distribution, the p-value is calculated). This will give us a framework to later obtain the same trade off when the distributional assumptions are violated.

In equation \ref{beta_definition2}, we described this trade off, $\beta(\hat{\alpha}) = F_{X_a}(S_{X_0}^{-1}(\hat{\alpha}))$. Let's see what this looks like for the rate test. We have the following corollary to theorem \ref{rate_test}:

\begin{corollary}
Given the null hypothesis for the rate test in theorem \ref{rate_test} $X_0 \overset{H_0}{\sim} \text{B}(n,\frac{t_1}{t_1+t_0})$, $N_0$ describing the number of events in the control group (in observation time, $t_0$) and $N_1$ describing the same for the treatment group in observation time $t_1$, we get the false negative rate corresponding to a type-1 error rate of $\hat{\alpha}$:

$$\beta(\hat{\alpha}) = \sum\limits_{n=0}^{\infty}\sum\limits_{j=0}^{S_{X_0}^{-1}(\hat{\alpha})}
P(N_0=j)P(N_1=n-j)$$

and the false negative rate corresponding to false positive rate $\alpha$:

$$\beta(\alpha) = \sum\limits_{n=0}^{\infty}\sum\limits_{j=0}^{S_{Y_0}^{-1}(\alpha)}
P(N_0=j)P(N_1=n-j)$$
\end{corollary}
\begin{proof}
Since our test conditions on the total number of events observed, $n$, we start with describing our $\beta$ under that condition as well. Denoting by $N_0$ and $N_1$ the number of events observed in groups 0 and 1 in observation times $t_0$ and $t_1$ respectively and noting that $X_a$, being the number of events from group-1 conditional on $n$ is a discrete random variable equation \ref{beta_def} becomes:

\begin{equation*}
\beta^{(n)}(\hat{\alpha}) = \sum\limits_{j=0}^{S_{X_0}^{-1}(\hat{\alpha})} P(X_a=j|N_0+N_1=n)
\end{equation*}

Since our test statistic for this particular test is simply the number of events from the first process we get $X_a \sim N_0$, making the equation above:

\begin{equation*}
\beta^{(n)}(\hat{\alpha}) = \sum\limits_{j=0}^{S_{X_0}^{-1}(\hat{\alpha})} P(N_0=j|N_0+N_1=n)
\end{equation*}

To get the overall $\beta$, we simply marginalize over all possible values of $n$ to get:

\begin{align}
\beta(\hat{\alpha}) =\sum\limits_{n=0}^{\infty}P(N_0+N_1=n) \beta^{(n)}(\hat{\alpha}) \nonumber \\
=\sum\limits_{n=0}^{\infty}P(N_0+N_1=n) \sum\limits_{j=0}^{S_{X_0}^{-1}(\hat{\alpha})} P(N_0=j|N_0+N_1=n)\label{beta_rate_1} \\
=\sum\limits_{n=0}^{\infty}P(N_0+N_1=n) \sum\limits_{j=0}^{S_{X_0}^{-1}(\hat{\alpha})} \frac{P(N_0=j\&N_0+N_1=n)}{P(N_0+N_1=n)} \nonumber \\
= \sum\limits_{n=0}^{\infty}\sum\limits_{j=0}^{S_{X_0}^{-1}(\hat{\alpha})}
P(N_0=j)P(N_1=n-j)\label{beta_rate_2}
\end{align}
\end{proof}
It is sometimes convenient to use equation \ref{beta_rate_1} (especially when the conditional distribution in that equation has a nice closed form) and other times, \ref{beta_rate_2}. Under $H_0$ (the assumptions of the rate test under the null hypothesis), $N_0$ and $N_1$ follow Poisson distributions with the same means, $\lambda t_0$ and $\lambda t_1$ respectively. The second part of the proposition follows as a result of equation \ref{beta_general}.

\begin{proposition}
If we apply the uniformly most powerful rate test as described in theorem \ref{rate_test} to $H_0$ defined as both treatment and control groups following $PP(\lambda)$ and $H_a$ defined as control following $PP(\lambda)$ and treatment following $PP(\lambda + \delta \lambda)$ where $\delta \lambda > 0$, the false negative rate for any false positive rate goes to zero if we collect data from both processes for a very large period of time ($t \to \infty$).
\end{proposition}
\begin{proof}
We will prove this for the special case, $\alpha = \frac 1 2$. 
Let's assume that both groups (control and treatment) are observed for a time period, $t$.

Substituting the results of lemmas \ref{sum_of_poissons} and \ref{lem:conditional_binom} into equation \ref{beta_rate_1} we get:

\begin{multline}\beta(\alpha) = \sum\limits_{n=0}^\infty \frac{e^{-(2\lambda+\delta \lambda)t} ((2\lambda+\delta \lambda)t)^n}{n!} \\ \sum\limits_{j=0}^{S_{X_0}^{-1}(\hat{\alpha})}{n \choose j}\left(\frac{\lambda+\delta \lambda}{2\lambda+\delta \lambda}\right)^j \left(\frac{\lambda}{2\lambda+\delta \lambda}\right)^{n-j}\end{multline}

Where $X_0 \overset{D}{=} B(n,\frac{1}{2})$

$$=e^{-(2\lambda+\delta \lambda)}\sum\limits_{n=0}^\infty \frac{(\lambda t)^n}{n!}\sum\limits_{j=0}^{S_{X_0}^{-1}(\alpha)} {n \choose j} \left(1+\frac{\delta \lambda}{\lambda}\right)^j$$

We will proceed from here for the special case, $\alpha=\frac{1}{2}$. This makes $S_{X_0}^{-1}(\frac 1 2) = \left[\frac{n}{2}\right]$ (where $[n]$ is the greatest integer $\leq n$). So we get:

\begin{equation}\label{beta_rate_on_poisson}\beta \left(\frac{1}{2}\right) = e^{-(2\lambda+\delta \lambda)t}\sum\limits_{n=0}^\infty \frac{(\lambda t)^n}{n!}\sum\limits_{j=0}^{\left[\frac{n}{2}\right]} {n \choose j} \left(1+\frac{\delta \lambda}{\lambda}\right)^j\end{equation}

Now if we show for some $\eta > 0$, 
\begin{equation}\label{beta_rate_on_poisson_lt}\beta \left(\frac{1}{2}\right) < e^{-(2\lambda+\delta \lambda)t}\sum\limits_{n=0}^\infty \frac{(\lambda t)^n}{n!}\left(2+\frac{\delta \lambda}{\lambda}-\eta\right)^n\end{equation}

we would have shown the result since the Taylor's expansion of $e^x$ implies:

$$\beta \left(\frac{1}{2}\right)  < e^{-(2\lambda+\delta \lambda)t}  e^{(2\lambda+\delta \lambda-\eta)t} = e^{-\eta t}$$

And so (in conjunction with the fact that $\beta(\alpha)\geq0 \;\; \forall \alpha \in [0,1)$),

$$\lim_{t \to \infty} \beta \left(\frac{1}{2}\right)  = 0$$

Comparing equations \ref{beta_rate_on_poisson} and \ref{beta_rate_on_poisson_lt}, the inequality would certainly hold if it were possible to find an $\eta$ such that:

$$
\left(2+\frac{\delta \lambda}{\lambda}-\eta\right)^n > \sum\limits_{j=0}^{\left[\frac{n}{2}\right]} {n \choose j} \left(1+\frac{\delta \lambda}{\lambda}\right)^j \;\; \forall n
$$
Let $p=\frac{\delta \lambda}{\lambda}$ and this requirement becomes:

\begin{equation}\label{eta_req}\eta < (2+p) - \left( \sum\limits_{j=0}^{\left[\frac{n}{2}\right]} {n \choose j} \left(1+p\right)^j \right)^{\frac{1}{n}}\end{equation}

This is obviously true for any finite value of $n>1$ since the summation, $\left( \sum\limits_{j=0}^{n} {n \choose j} \left(1+p\right)^j \right)^{\frac{1}{n}}=2+p$ and the summation in equation \ref{eta_req} is missing some positive terms compared with this summation. Those terms will sum to something finite and allow us to choose some $\eta>0$. This holds for all $p>-1$.

The only concern remaining is that we might not be able to find an $\eta>0$ satisfying equation \ref{eta_req} when $n \to \infty$. And indeed, this turns out to be the case only for $p>0$.

Let's find the limit:

\begin{align*}L = \lim_{n \to \infty} \left( \sum\limits_{j=0}^{\left[\frac{n}{2}\right]} {n \choose j} \left(1+p\right)^j \right)^{\frac{1}{n}}\\
=\lim_{n \to \infty} \left(\sum\limits_{j=0}^n {2n \choose j}(1+p)^j\right)^\frac{1}{2n}
\end{align*}

Noting the inequality

$$ \binom{2n}{n}(1+p)^n \leq \sum_{j=0}^{n} \binom{2n}{j}(1+p)^j \leq n \cdot \binom{2n}{n}(1+p)^n $$
and the limit $\lim_{n\to\infty} n^{1/n} = 1$, we deduce that

$$ L = \lim_{n\to\infty} \left[ \binom{2n}{n} (1+p)^n \right]^{\frac{1}{2n}} = 2\sqrt{1+p}. $$

AM-GM inequality on $1$ and $1+p$ guarantees that $L \leq 2+p$ and the equality holds if and only if $p = 0$. Hence we see that an $\eta>0$ satisfying equation \ref{eta_req} will exist if $p>0$ but not if for example, $p = 0$. This shows an $\eta$ exists for the case we're interested in ($\delta \lambda>0$ and hence $p>0$) and concludes the proof.
\end{proof}

\section{Breaking the test}
We now have a pretty straightforward test for testing the rates of two point processes which is indeed proven to be the best possible when these are Poisson point processes. All we need is four numbers, the number of events and time period of observation in which those events were collected for two groups. But is this too simple? The Poisson point process is quite restrictive in the assumptions it makes and is almost never a good model for real-world data. Is applying a test built on it's assumptions then, naive? Let's explore this question in this section by breaking every possible underlying assumption and investigating how the test behaves.

\subsection{Swapping out the distribution of the null hypothesis}

In the construction of our hypothesis test, we used equation \ref{eqn:binom_cond}, which allowed us to condition on $n$ and use the fact that the distribution of the number of events from the second process, $n_1$, is Binomial (let's call it $X_0$). Similarly, the distribution of our test statistic, $n_1$ under the alternate hypothesis (given some effect size) is $X_a$ which happens to also be Binomial with the same number of tosses, $n$ parameter but a different probability of heads parameter, $p$.

In the spirit of finding ways to break our test, let's say we won't be using the distribution of the null hypothesis, $X_0$ anymore and will instead swap it out with another arbitrary distribution, $Z_0$, with the same support as $X_0$ (non-negative integers $\leq n$). The following result is somewhat surprising:

\begin{theorem}\label{y0_no_effect}
For any one-sided hypothesis test, if we swap out the distribution of the null hypothesis, $X_0$ with another arbitrary distribution, $Z_0$ that has the same support, we get the same false negative rate corresponding to any false positive rate.
\end{theorem}
\begin{proof}
From equations \ref{hat_alpha} and \ref{beta_definition2}, we get the false positive rate to false negative rate trade-off function.

$$\beta(\alpha) = \beta(\hat{\alpha}) = F_{X_a}(S_{X_0}^{-1}(\alpha))$$

Now, consider the test where we replace $X_0$ with $Z_0$ (known henceforth as the ``contorted test''). First, let's obtain a relationship for the false positive rate for this test, $\alpha'(\hat{\alpha})$. Using a similar reasoning as we used to obtain equation \ref{hat_alpha},

\begin{align}
\alpha'(\hat{\alpha})= P(S_{Z_0}(X_0)<\hat{\alpha})\nonumber \\
=P(X_0>S_{Z_0}^{-1}(\hat{\alpha})) \nonumber \\
=S_{X_0}S_{Z_0}^{-1}(\hat{\alpha}) \label{hat_alpha_replace_x}
\end{align}

Note that this time, the two functions don't cancel out. So, the type-1 error ($\hat{\alpha}$) for our contorted test (with $X_0$ replaced with $Z_0$) is different from the false positive rate, $\alpha'(\hat{\alpha})$.

Now, let's explore the false negative rate of this contorted test. Using a similar approach as for proposition \ref{prop:fnr_def} we get:

\begin{align}
\beta'(\hat{\alpha}) = P(S_{Z_0}(X_a)>\hat{\alpha}) \nonumber\\
= P(X_a<S_{Z_0}^{-1}(\hat{\alpha})) \nonumber\\
= F_{X_a}S_{Z_0}^{-1}(\hat{\alpha})\label{beta_replace_x}
\end{align}

In equation \ref{hat_alpha_replace_x}, applying $S_{X_0}^{-1}$ followed by $S_{Z_0}$ to both sides we get,

$$ \hat{\alpha} = S_{Z_0}S_{X_0}^{-1}(\alpha')$$

And substituting this into equation \ref{beta_replace_x} we get the $\beta$-$\alpha$ trade off for this test:

$$\beta'(\alpha') = F_{X_a}S_{Z_0}^{-1}S_{Z_0}S_{X_0}^{-1}(\alpha')$$
$$ = F_{X_a}S_{X_0}^{-1}(\alpha')$$

Which means for the contorted test, given a false positive rate $\alpha$, the false negative rate $\beta'(\alpha)$ is

$$\beta'(\alpha) =  F_{X_a}S_{X_0}^{-1}(\alpha)$$

But the above is the same as the $\beta(\alpha)$ we got from equation \ref{beta_definition2}. This shows that given a false positive rate $\alpha$, the false negative rates for the two tests, $\beta(\alpha)$ and $\beta'(\alpha)$ are equal and proves the theorem. It is also easy to see that we could have replaced $X_0$ with $Y_0$ and $X_a$ with $Y_a$ and reached the same conclusion, meaning the theorem continues to hold even when the original and contorted tests are applied to data that doesn't follow the assumptions of the test.
\end{proof}

Consider we're trying to apply the hypothesis test for failure rates described in theorem \ref{rate_test} to point processes that aren't Poisson processes. One consequence of this violation of the distributional assumption would be that the conditional (on the total number of events, $n$) distribution of the test statistic, $n_1$ will no longer be Binomial. We might consider trying to find what this distribution is and replace the Binomial distribution with it so as to devise a test more tailored to the point processes from our data. Per theorem \ref{y0_no_effect}, this would be a waste of time as far as the $\beta$-$\alpha$ trade off goes as swapping out the Binomial with any other distribution under the sun would not improve the false negative rate we get corresponding to a false positive rate. 

Also note that nothing in the derivation was specific to the rate tests. The conclusion of theorem \ref{y0_no_effect} holds for any one-sided hypothesis test. In the famous two sample t-test for comparing means for instance, if we swap out the t-distribution with a normal or even some strange multi-modal distribution, the false negative to false positive trade off will remain unchanged.

\subsection{Violating assumptions}
Now, we get to scenarios where we apply the rate test as described in theorem \ref{rate_test} as-is to point processes that are not Poisson processes. For example, a core property of the Poisson point process is that the mean and variance of the count of events within any interval are the same. Many real world point processes don't depict this behavior, with variance typically being higher than mean. 

What then, is the price we pay in still applying the rate test derived on the assumptions of the Poisson process to rates from these non-Poisson processes? This depends of course, on the particular point process we're dealing with. In this section, we'll consider different ways we can generalize the Poisson point process with its constant failure rate and then see what happens with the rate test applied to them. Three of these generalizations are covered in section 5.4 of \cite{ross} viz the non-homogeneous, compound and mixed Poisson processes. For a non-homogeneous Poisson process, the rate is allowed to vary with time ($\lambda(t)$), but in a way that it isn't affected by the arrivals of events. The number of events within any intervals is still Poisson distributed in this process (with mean $\int \lambda(t) dt$) and so, doesn't depart from the Poisson process in a significant way. The other two generalizations do fundamentally alter the distribution of the number of events within intervals and we'll deal with them in turn.

\subsubsection{The Compound Poisson Process}
The Compound Poisson process, covered in section 5.4.2 of [1] involves a Compounding distribution superposed on the Poisson process. We still have a Poisson process dictating event arrivals. However, each time we get an arrival from the Poisson process, we get a random number of events (the compounding random variable, $C$) instead of a single event.

This is especially relevant to failures within a cloud platform like Microsoft Azure wherein there are multiple single points of failure that have the effect of clustering machine reboots together, leading to a higher variance of event counts within time interval than mean. The most obvious one is multiple virtual machines (the units rented to customers; VMs) being co-hosted on a single physical machine (or node). If the node goes down, all the VMs will go down together. Now, the number of VMs on a node when it goes down will be a random variable itself (the compounding random variable). 

Per equation (5.23) of \cite{ross} (or simply from the definition), the number of events in any interval, $t$ will be given by ($C_j$ are independent identically distributed with the same distribution as $C$):

$$M(t) = \sum\limits_{j=1}^{N(t)} C_j$$

Per equations (5.24) and (5.25) from \cite{ross}, the mean and variance of such a point process will become (assuming the underlying Poisson process has a rate, $\lambda$):

$$E[M(t)] = \lambda t E[C]$$
and,
$$ V[M(t)] = \lambda t E[C]^2 $$

This allows for our variance to be much higher than the mean and makes clear the fact that the number of events in any interval is no longer Poisson distributed.

\paragraph{Deterministic-ally compounded Poisson process}
What is the simplest kind of compounding we can do (apart from none at all)? We can have a constant number of events for each Poisson arrival. In other words, $C$ becomes a deterministic number instead of a random variable and let's say the value it takes each time is $l$ (we'll call such a process $DP(\lambda,l)$). For such a process, the number of events generated by either group must be an integer multiple of $l$. Also, per equations (5.24) and (5.25) of \cite{ross}, the mean and variance in the number of events become:

$$E[M(t)] = \lambda t l$$

$$V[M(t)] = \lambda t l^2$$

Since the variance is now higher than the mean, this is a fundamentally different point process from the Poisson point process.

\begin{lemma}
For the deterministic-ally compounded Poisson process, the probability mass function of the point process becomes:
\begin{equation}\label{d_pois_pmf}
P(M(t)=n) =
	\begin{cases}
      P(N(t)=k), & \text{if}\ n=lk \\
      0, & \text{otherwise}
    \end{cases}
\end{equation}
\end{lemma}
\begin{proof}
Since every Poisson arrival results in exactly $l$ events, the number of events in any interval $t$ must be a multiple of $l$. And if we observe $lk$ events in any interval, then the number of Poisson arrivals must have been $k$.
\end{proof}

\begin{lemma}\label{properties_det_cm_poisson}
Let $Y_0$ be the distribution of the number of events in the treatment group conditional on the total events across both groups being $n$. We must have:

\begin{itemize}
\item{$n=kl \;\; \exists \;\; k \in \mathbb{Z}$}
\item{$Y_0 \overset{D}{=} l B\left(k,\frac{t_1}{t_1+t_0}\right)$}
\item{$S_{Y_0}^{-1}(\alpha) = l S_{X_0}^{-1}(\alpha)$ where $X_0 \overset{D}{=} B\left(k,\frac{t_1}{t_1+t_0}\right)$}
\end{itemize}
\end{lemma}
\begin{proof}
For the first part, since the events from any $DP(\lambda,l)$ must be a multiple of $l$, so too must be the number of events from a sum of two of them.

For the second part, if $n=kl$ we can surmise that the number of Poisson arrivals across both groups was $k$. So, the conditional Poisson arrivals from the treatment groups is still governed by the conclusions of lemma \ref{lem:conditional_binom} and corollary \ref{null_rate_test}. And once we know the Poisson arrivals from the treatment group, the total failures will just be $l$ times that.

For the third part, the probability mass for any $k \in \mathbb{Z}$ under $X_0$ is simply moved to $kl$ under $Y_0$. Hence, if $k$ is the point where the sum of probabilities after it sum to $\alpha$ under $X_0$, this point will get scaled by $l$ as well under $Y_0$.
\end{proof}

Now, let's assume that the number of failures in our two groups follow the deterministic-ally compounded Poisson process.

Under $H_0'$, we will have both treatment and control groups following $DP(\lambda,l)$ and under $H_a'$, we add an effect size to the Poisson failure rate to the treatment group. So while the control group still follows $DP(\lambda,l)$ the treatment group now follows $DP(\lambda+\delta \lambda,l)$.

As before, $H_0$ involves both groups following a Poisson process, $PP(\lambda)$ and $H_a$ involves the control group following $PP(\lambda)$ and treatment following $PP(\lambda+\delta \lambda)$. We have the distributional assumptions for the test statistic (number of failures from treatment group conditional on total failures being $n$):

$$X \overset{H_0}{\sim} X_0 \overset{D}{=} B\left(n,\frac{t_1}{t_1+t_0}\right)$$
$$X \overset{H_a}{\sim} X_a \overset{D}{=} B\left(n,\frac{(\lambda+\delta \lambda)t_1}{(\lambda+\delta\lambda)t_1+\lambda t_0}\right)$$
$$X \overset{H_0'}{\sim} Y_0 \overset{D}{=} l \times B\left(n,\frac{t_1}{t_1+t_0}\right)$$
$$X \overset{H_a'}{\sim} Y_a \overset{D}{=} l \times B\left(n,\frac{(\lambda+\delta \lambda)t_1}{(\lambda+\delta\lambda)t_1+\lambda t_0}\right)$$

\begin{proposition}\label{deter_cmpd_same_power}
If we apply the rate test defined in theorem \ref{rate_test} to data generated under $H_0'$ as the null and $H_a'$ as the alternate hypothesis (defined above), the false negative rate for any given false positive rate is exactly the same as when we apply the test to data generated from $H_0$ as the null and $H_a$ as the alternate.
\end{proposition}
\begin{proof}
From equation \ref{beta_rate_2}, we get under $H_0'$ and $H_a'$ the false negative rate as a function of the false positive rate:

$$\beta_M(\alpha) = \sum\limits_{n=0}^{\infty}\sum\limits_{j=0}^{S_{Y_0}^{-1}(\alpha)}
P(M_1(t)=j)P(M_0(t)=n-j)$$

From the third result of lemma \ref{properties_det_cm_poisson}:

$$= \sum\limits_{n=0}^{\infty}\sum\limits_{j=0}^{l S_{X_0}^{-1}(\alpha)}
P(M_1(t)=j)P(M_0(t)=n-j)$$

Setting $n=ml$, noting that both the probabilities in the summation will be $0$ when $n$ is not a multiple of $l$ (per equation \ref{d_pois_pmf}):

$$= \sum\limits_{m=0}^{\infty}\sum\limits_{j=0}^{l S_{X_0}^{-1}(\alpha)}
P(M_1(t)=j)P(M_0(t)=ml-j)$$

Similarly setting $j=kl$ and applying equation \ref{d_pois_pmf} (assuming $N_0$ and $N_1$ follow $PP(\lambda)$ and $PP(\lambda+\delta \lambda)$ respectively):

$$\beta_M(\alpha)= \sum\limits_{m=0}^{\infty}\sum\limits_{k=0}^{S_{X_0}^{-1}(\alpha)}
P(N_1(t)=k)P(N_0(t)=m-k)$$

Also from equation \ref{beta_rate_2} , the same rate under $H_0$ and $H_a$:

$$\beta_N(\alpha) = \sum\limits_{n=0}^{\infty}\sum\limits_{j=0}^{S_{X_0}^{-1}(\alpha)}
P(N_1(t)=j)P(N_0(t)=n-j)$$

It is easy to see that $\beta_M(\alpha)=\beta_N(\alpha)$. Basically, in the summation for $\beta_M$, we simply shifted all the terms from the $\beta_N$ summation to multiples of $l$, but still ended up summing all the same terms. That is why the false negative rate remained exactly the same.
\end{proof}

Does getting the same false negative rate mean somehow that the deterministic-ally compounded Poisson process is similar to the regular Poisson process? One way to visualize how different two distributions are is with a quantile-quantile plot. A quantile of a distribution, $X$ is defined as its inverse CDF, $F_X^{-1}(q)$. If we take another distribution, $Y$ and plot $F_{Y}^{-1}(q)$ with $F_X^{-1}(q)$ for various values of $q$, we get the Q-Q plot which should be a straight line if the two distributions are the same. The idea is that we're plotting:

$$q = F_X(x) = F_Y(y)$$

or in other words,

$$y = F_Y^{-1}F_X(x)$$

If $Y \overset{D}{=}X$, the QQ plot will become the line: $y=x$. So, the closer the distributions are, the closer the plot is to the line $y=x$. Compare this to the relationship between false positive rate and type-1 error from equation \ref{alpha_general}: $\alpha(\hat{\alpha}) = S_{X_0}S_{Y_0}^{-1}(\hat{\alpha})$. Again, we expect that if $X_0 \overset{D}{=} Y_0$, we will get $\alpha=\hat{\alpha}$, a straight line. For the deterministically compounded Poisson process, we can get $\alpha(\hat{\alpha})$ in closed form.

\begin{proposition}
When the rate test from theorem \ref{rate_test} is applied to the deterministically compounded Poisson process, we get a false positive rate ($X_0\overset{D}{=}B(ml,\frac{t_1}{t_1+t_0})$):

\begin{multline*}\alpha(\hat{\alpha}) = 1-\\ \sum\limits_{m=0}^{\infty} \frac{e^{-\lambda(t_1+t_2)}(\lambda(t_1+t_2))^m}{m!} \sum\limits_{k=0}^{\left[ \frac{S_{X_0}^{-1}(\hat{\alpha})}{l} \right]} {m \choose k}p^k(1-p)^{m-k}\end{multline*}

Where $p=\frac{t_1}{t_1+t_0}$.
\end{proposition}
\begin{proof}
Per equation \ref{alpha_general}, we have $\alpha(\hat{\alpha}) = S_{Y_0}S_{X_0}^{-1}(\hat{\alpha})$. So conditional on a total of $n$ events from both groups, the false positive rate becomes:

\begin{align}
\alpha^{(n)}(\hat{\alpha}) = \sum\limits_{j=S_{X_0}^{-1}(\hat{\alpha})+1}^{n}P(N_0=j|N_0+N_1=n)\nonumber\\
= 1-\sum\limits_{j=0}^{S_{X_0}^{-1}(\hat{\alpha})}P(N_0=j|N_0+N_1=n)\nonumber
\end{align}

Where $N_0$ and $N_1$ follow the distributional assumptions of $H_0'$, which is deterministically compounded Poisson. Marginalizing over all $n$, we get:

\begin{equation*}
\alpha(\hat{\alpha}) = 1-\sum\limits_{n=0}^{\infty}\sum\limits_{j=0}^{S_{X_0}^{-1}(\hat{\alpha})}P(N_0=j|N_0+N_1=n)P(N_0+N_1=n)
\end{equation*}

Since $N_0$ and $N_1$ follow the deterministically compounded Poisson process and $Y_0 \overset{D}{=} l \times B\left(n,\frac{t_1}{t_1+t_0}\right)$, we will only get non-zero terms in this summation when $n=ml$ and $j=kl$. Also, although $j$ will go up to $S_{X_0}^{-1}(\hat{\alpha})$, it'll only hit non-zero terms until $j=\left[\frac{S_{X_0}^{-1}(\hat{\alpha})}{l}\right]$

This gives us the result:

\begin{multline}\label{alpha_det_comp_pois}
\alpha(\hat{\alpha}) = 1-\sum\limits_{m=0}^{\infty} \frac{e^{-\lambda(t_1+t_2)}(\lambda(t_1+t_2))^m}{m!} \\ \sum\limits_{k=0}^{\left[ \frac{S_{X_0}^{-1}(\hat{\alpha})}{l} \right]} {m \choose k}p^k(1-p)^{m-k}
\end{multline}
Where,
$$p=\frac{t_1}{t_1+t_0}$$
\end{proof}
The reason we don't get the same terms unlike with $\beta(\alpha)$ is that the inner summation is now up to $S_{X_0}^{-1}(\hat{\alpha})$ instead of $S_{Y_0}^{-1}(\hat{\alpha})$. We can now use equation \ref{alpha_det_comp_pois} or simulation to obtain plots for $\alpha$ with $\hat{\alpha}$ (shown below). We see that as $l$ increases, the deterministically compounded distribution starts moving away from the Poisson process as expected. However, this deviation has no effect at all on the false negative to false positive trade-off as demonstrated earlier, which stays the same regardless of $l$.  We do note that our false positive rate starts diverging from the type-1 error rate, but this can easily be corrected for, especially since we know the closed form. This should drive home the point that violating the distributional assumptions of the null hypothesis of a hypothesis test isn't a good argument for not applying it or necessarily a sign of weakness for it.

In the next distributional violation we cover, we'll take this a step further and see that violating the assumptions of the test to a larger extent actually causes the false negative to false positive trade off to become better!

\begin{figure}
  \includegraphics[width=0.8\linewidth]{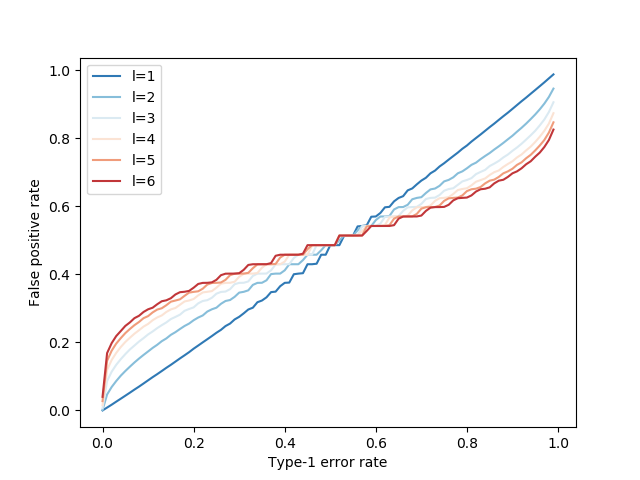}
  \caption{Relationship between type-1 error rate and false positive rates for the deterministically compounded Poisson process with various values of compounding.}
  \label{alpha_alpha_hat_deter_cmpd_poisson}
\end{figure}

\paragraph{Binomially compounded Poisson process}
Next, we explore the Binomially compounded Poisson process where the compounding factor, $C$ follows a Binomial distribution with parameters $l$ and $p$. Meaning with each arrival, instead of there being just one event, we generate a binomial random variable (parameters $l$ and $p$) and that dictates the number of events. Let $M(t)$ be this compound point process and $N(t)$ the underlying Poisson process. We then get:

$$P(M(t) = j) = \sum\limits_{m=0}^\infty P(M(t)=j|N(t)=m)P(N(t)=m)$$

Conditional on $N(t)=m$, we basically end up summing $m$ binomial random variables with the same $p$ parameter. This becomes another binomial with parameters $lm$ and $p$. So we get:

$$P(M(t)=j)=\sum\limits_{m=0}^\infty {lm \choose j}p^j(1-p)^{lm-j} \left(\frac{e^{-\lambda t} (\lambda t)^m}{m!}\right)$$

Unfortunately, there doesn't seem to be a closed form for this summation, meaning we can't use equation \ref{beta_rate_2} to get a closed form for the $\alpha$-$\beta$ trade off either. We can however get it from simulation.

When we plot the $\beta$-$\alpha$ trade off for $p=0.7$ and various values of $l$, we see that applying the rate test to the Binomially compounded Poisson process does give us a worse $\alpha$-$\beta$ trade off curve than when it is applied to the Poisson process as expected. However, as we increase $l$, the trade off curve actually becomes better and starts approaching the one for the Poisson process.

\begin{figure}
  \includegraphics[width=0.8\linewidth]{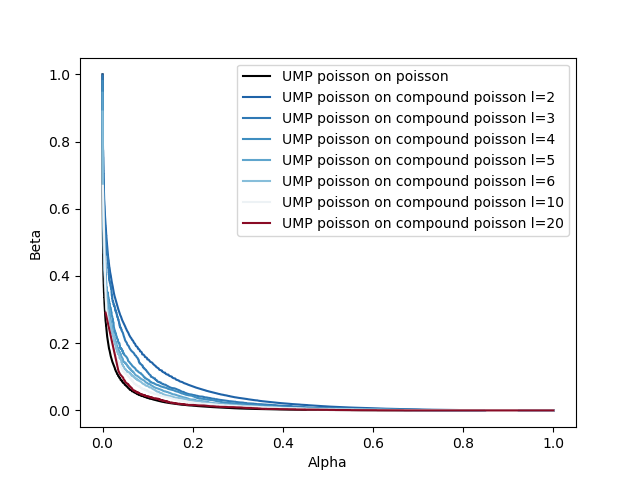}
  \caption{$\alpha$-$\beta$ trade off for Binomially compounded Poisson for $p=0.7$ and various values of $l$.}
  \label{alpha-beta_cmpd_binom}
\end{figure}

On the other hand, when we plot the false positive rate with type-1 error rate (which as mentioned earlier, tells us how close the distributions are to each other like a QQ plot), we see that higher values of $l$ diverge more from the $\alpha=\hat{\alpha}$ line, meaning they get more ``different'' as distributions. See figure \ref{alpha_alpha_hat_deter_cmpd_poisson}.

\begin{figure}
  \includegraphics[width=0.8\linewidth]{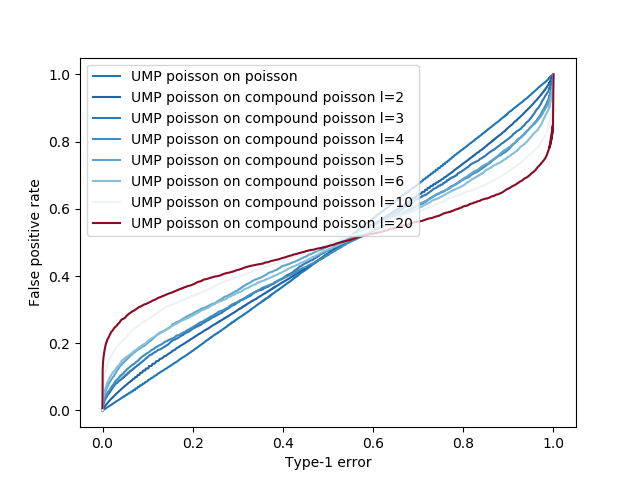}
  \caption{$\alpha$ vs $\hat{\alpha}$ for the Binomially compounded Poisson process for various values of $l$.}
  \label{alpha_alpha_hat_deter_cmpd_poisson}
\end{figure}

So, we see from figures \ref{alpha-beta_cmpd_binom} and \ref{alpha_alpha_hat_deter_cmpd_poisson} that as we increase the $l$ parameter, the distribution diverges more from Poisson while conversely, the false negative rate becomes better and approaches that of the Poisson distribution. What could be the reason for this behavior? We know that the mean of a $B(l,p)$ distribution is $lp$ while its variance is $lp-lp^2$. So as we increase $l$, the difference between the mean and variance increases, making the binomially compounded Poisson closer and closer to the deterministically compounded Poisson with compounding factor, $[lp]$. And we showed in proposition \ref{deter_cmpd_same_power} that the deterministically compounded Poisson process gets the same $\alpha$-$\beta$ trade off as a regular Poisson process.

To validate our assumption that the binomially compounded Poisson approaches the deterministically compounded Poisson as the $l$ parameter of the $B(l,p)$ distribution increases, we plot a QQ plot between them in figure \ref{qq_deter_poisson_binom_poisson}. And we see indeed that the plot becomes closer and closer to a straight line as $l$ increases.

\begin{figure}
  \includegraphics[width=0.8\linewidth]{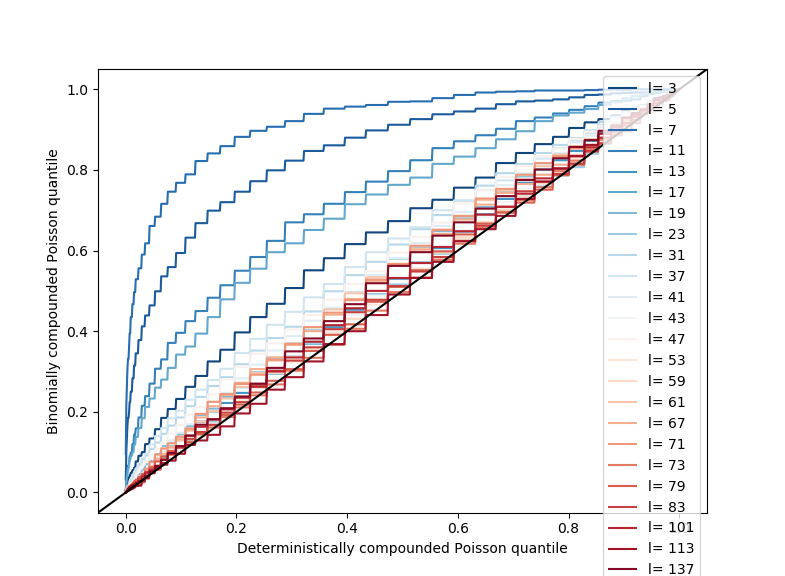}
  \caption{QQ plot between deterministically and binomially compounded Poisson. The $l$ is number of binomial tosses and $p=0.7$}
  \label{qq_deter_poisson_binom_poisson}
\end{figure}

In conclusion, we saw for the Binomially compounded Poisson process that as it diverges more from the distributional assumptions of the rate test, the false negative rate corresponding to any false positive rate actually improves. Also note that since the Binomial distribution is a very reasonable assumption for compounding processes that involve a random numbers generated over integers $\leq n$, this is not a pathological edge case.

Hopefully these examples with different Compound Poisson processes have demonstrated that violating the assumptions of the null hypothesis on top of which the test is built isn't necessarily detrimental to its effectiveness. The next generalization of the Poisson process we'll cover is the mixed or conditional Poisson distribution, but that deserves a section of its own.

\section{Mixed Poisson process: starring the Negative Binomial}
As we know, the Poisson process has only one parameter, the rate $\lambda$. In generalizing the Poisson process then, we could condition on this parameter, imagining that it itself is picked from some distribution (say $L$). So conditional on $L=\lambda$, we get a Poisson process with rate $\lambda$. This is called the ``mixed Poisson process'' or ``conditional Poisson process'' and is covered in section 5.4.3 of \cite{ross} as well as \cite{mxd_poisson_paper}. 

Although any distribution can be chosen for the rate $L$, a natural choice is the Gamma distribution since it is a \href{https://en.wikipedia.org/wiki/Conjugate_prior}{conjugate prior}, meaning that if we choose to update $L$ once we observe some data, it will also be a (probably different) Gamma distribution. This also means that we have closed form expressions for the distributions of the number of events in interval $t$ ($N(t)$) and inter-arrival time between events ($T$) associated with the point process. Example 5.29 of \cite{ross} shows that if $L$ is Gamma distributed with parameters $\theta$ and $m$, implying a density function:

\begin{equation}\label{gamma_definition}g_L(\lambda) = \theta e^{-\theta \lambda}\frac{(\theta \lambda)^{m-1}}{(m-1)!}\end{equation}
then the distribution of $N(t)$ becomes negative binomial with probability mass function:

\begin{equation}\label{eqn:neg_binom}
P(N(t)=n)={n+m-1 \choose n} \left(\frac{\theta}{\theta+t}\right)^m \left(\frac{t}{\theta+t}\right)^n 
\end{equation}

This is the number of tails required before we observe a total of $m$ heads when the probability of heads is $p=\frac{\theta}{t+\theta}$.

Henceforth, we will denote this particular mixed Poisson process by $NBP(m,\theta)$ and the distribution by $NB(m,p)$. Since the negative binomial distribution is the most common substitute to address the property of the Poisson distribution wherein its mean and variance are equal, the negative binomial process likewise addresses this and allows the variance in $N(t)$ to be greater than the mean (which is a feature we see in most real world point processes like Azure reboots, natural disasters, machine and human organ failures, etc.). 

We can also obtain a closed form expression for the probability density function of the inter-arrival times ($T$):

\begin{proposition}
For a mixed Poisson process with the rate being drawn from $L$, a Gamma distribution; the distribution of the inter-arrival times, $T$ becomes Lomax (Pareto type-2).
\end{proposition}
\begin{proof}
Since conditional on $L=\lambda$ we get a regular Poisson point process with inter-arrival times being exponentially distributed with rate $\lambda$, we marginalize out $\lambda$ to get the new inter-arrival distribution for the mixed Poisson process:

$$f_T(t) = \int\limits_0^\infty (\lambda e^{-\lambda t}) g_L(\lambda) d\lambda$$

Substituting the Gamma density from equation \ref{gamma_definition}:
$$f_T(t) = \int\limits_0^\infty (\lambda e^{-\lambda t}) \theta e^{-\theta \lambda} \frac{(\theta \lambda)^{m-1}}{(m-1)!} d\lambda$$
$$=\frac{\theta^m}{(m-1)! } \int\limits_0^\infty e^{-(t+\theta)\lambda} \lambda^m d \lambda$$
Now,

$$\int\limits_0^\infty e^{-(t+\theta)\lambda }\lambda^m d \lambda = \frac{m!}{(t+\theta)^{m+1}}$$

Which gives us:

\begin{equation}\label{lomax_def}f_T(t) = \left(\frac{\theta}{t+\theta}\right)^m\end{equation}

And this is the PDF of the Lomax distribution.
\end{proof}

The Lomax distribution is polynomially instead of exponentially decaying, meaning it has much heavier tails than the exponential distribution. It is easy to see from equation \cite{haz_rate_def} that the Lomax distribution has a decreasing hazard rate and this is consistent with an over-dispersed point process (variance higher than mean). To see this, consider dividing the time over which we observe the process into tiny slices. If each interval is small enough, we'll either see zero or one events in them, making them Bernoulli random variables. For a Poisson process, these Bernoulli random variables were independent since the event rate at one of these intervals stays the same regardless of what happens around it. For the Lomax distribution with its decreasing hazard rate on the other hand, if a long time has passed without an event occurring, the decreasing event rate makes more events unlikely as well. This makes the Bernoulli variables positively correlated, making the variance in their sum (the total number of events in our observation interval) greater than for the Poisson, leading to over-dispersion. 

In fact, it can be shown that Poisson mixture models are only capable of modeling decreasing hazard rate inter-arrivals and hence over-dispersed point processes (the many examples covered in \cite{mxd_poisson_paper} for example all have this property). If we want to model an under-dispersed point process (variance lower than mean), we can choose an inter-arrival distribution that is capable of modeling increasing hazard rates. The Weibull distribution is one such candidate, capable of modeling both monotonically increasing and decreasing hazard rates and is covered in \cite{weibull_book}.

The Poisson process had the property of independent increments, meaning we can start an instance of the process and observe it for a long period of time or start many independent instances and observe them all for shorter periods. As long as the total time of observation is the same, our estimator for the failure rate has the exact same properties. Since for a mixed Poisson process, observing one interval of time gives us information about the mixing distribution and hence informs what will happen in proceeding intervals, we no longer have this independent increments property. 

While equation \ref{rate_def} is still an unbiased estimator for the average failure rate ($E(L)$), it turns out that it's better to collect many small intervals than one big one. When we observe the process for one long contiguous interval, we get an estimator that isn't asymptotically consistent (meaning the variance in the estimator doesn't converge to $0$ even as the observation period becomes arbitrarily large).

\begin{proposition}\label{prop:no_consistent}
If we hzve a single observation window for the mixed Poisson process mixed with distribution $L$, the estimator of the rate will become $E(L)$ and it's variance will be bounded below by $V(L)$, meaning it won't be asymptotically consistent.
\end{proposition}
\begin{proof}
It is a regular Poisson process when conditioned on some distribution $L(\lambda)$ of the rate, $\lambda$. Now, we're still interested in calculating the average hazard rate of this process. It's clear by definition (conditional on $L$, we get the regular Poisson process):

$$E(N(t)|L)=Lt$$
$$V(N(t)|L)=Lt$$

Using the law of total expectation:

$$E(N(t)) = t E(L)$$
Using the law of total variance:

$$V(N(t))=E(V(N(t)|L))+V(E(N(t)|L))$$
$$=E(Lt)+V(Lt)$$
$$=tE(L)+t^2V(L)$$

This means that if we observe this process for a large period of time, $t$, we can estimate the average hazard rate:

$$\hat{\lambda} = \frac{N(t)}{t} $$
$$=>E(\hat{\lambda})= E(L)$$

And the variance of this estimator becomes:

$$V(\hat{\lambda}) = \frac{V(N(t))}{t^2} = \frac{E(L)}{t}+V(L)$$
And so we have
$$\lim_{t \to \infty} V(\hat{\lambda}) \to V(L)$$.
\end{proof}

The result above makes sense since the rate, $\lambda$ itself is drawn from a distribution. If we restrict ourselves to one observation window, we only sample one $\lambda$ from this distribution. So, it makes sense that no matter how long we observe this instance for, the variance in the sampling of $\lambda$ itself is always there. This seems to suggest the following proposition: 

\begin{proposition}\label{prop:nbd_consistent}
If we sample $n$ time intervals and count the number of failures across all of them, the estimator from equation \ref{rate_def} is consistent as $n$ increases.
\end{proposition}
\begin{proof}
Suppose we sample $n$ intervals, $t_1, t_2, \dots t_n$ from a mixed Poisson process and observe $N(t_i)$ failures in the interval $t_i$. The estimator from equation \ref{rate_def} will become:

$$\hat{\lambda} = \frac{\sum\limits_{i=1}^n N(t_i)}{\sum\limits_{i=1}^n t_i}$$

Now, from the properties of the mixed Poisson process we have:

$$E(N(t_i)|L) = L t_i$$
$$V(N(t_i)|L) = L t_i$$

So we get:

$$E(\hat{\lambda}) =  \frac{\sum\limits_{i=1}^n E(N(t_i))}{\sum\limits_{i=1}^n t_i} =  \frac{\sum\limits_{i=1}^n E(L)t_i}{\sum\limits_{i=1}^n t_i}=E(L)$$

And the variance:

\begin{align}
V(\hat{\lambda}) =  \frac{\sum\limits_{i=1}^n V(N(t_i))}{\left(\sum\limits_{i=1}^n t_i\right)^2} \nonumber\\
=  \frac{\sum\limits_{i=1}^n V(L)t_i^2}{\left(\sum\limits_{i=1}^n t_i\right)^2}\nonumber\\
=  V(L) \frac{\sum\limits_{i=1}^n t_i^2}{\left(\sum\limits_{i=1}^n t_i^2\right)+\left(\sum\limits_{i\neq j} t_i t_j\right)}\nonumber
\end{align}
It's clear that:
$$\lim_{n \to \infty} V(\hat{\lambda}) \to 0$$
\end{proof}

\subsection{Hypothesis testing for the rate of the $NBP(m,\theta)$}
Now, let's consider that the point process for the treatment and control group are $NBP(m,\theta)$ for appropriate values of the parameters. For our null hypothesis, $H_0'$, we will assume that both treatment and control groups follow the same negative binomial process, $NBP(m,\theta)$. 

For $H_a'$, we want the treatment group to have a higher average rate, $E(L)$. Since the mean of the Gamma distribution is $E(L) = \frac{m}{\theta}$, we can achieve this by either increasing $m$ or decreasing $\theta$. Since the Gamma distribution is obtained by summing $m$ exponential distributions, each with rate $\theta$, the natural choice is to reduce $\theta$. This is also the approach used in \cite{zhu}. So under $H_a'$, we will assume that the treatment group follows $NBP(m,\theta-\delta \theta)$.

From equations \ref{beta_rate_2} and \ref{eqn:neg_binom}, we get the false negative rate as a function of the chosen type-1 error rate, $\hat{\alpha}$ (where $X_0 \overset{D}{=} B\left(n,\frac{t_1}{t_1+t_0}\right)$):

$$\beta(\hat{\alpha}) = \sum\limits_{n=0}^{\infty}\sum\limits_{j=0}^{S_{X_0}^{-1}(\hat{\alpha})} P(N_0=j)P(N_1=n-j)$$

\begin{multline}\label{eqn:false_neg_nbd}=>\beta(\hat{\alpha}) =  \sum\limits_{n=0}^{\infty}\sum\limits_{j=0}^{S_{X_0}^{-1}(\hat{\alpha})} \left({m+j-1\choose j}p_1^m (1-p_1)^j \right)\\ \left({m+n-j-1\choose n-j}p_0^m (1-p_0)^{n-j} \right)\end{multline}

where,

$$p_1 = \frac{\theta-\delta \theta}{\theta-\delta \theta +t}$$
$$p_0 = \frac{\theta}{\theta +t}$$

The summation above doesn't seem to have a closed form. As a result of proposition \ref{prop:no_consistent}, we can expect that it won't converge to $0$ even as the observation period, $t \to \infty$. On the other hand, if we increase the number of observation periods indefinitely, it should converge to $0$ per proposition \ref{prop:nbd_consistent}. And we do see this through numerically calculating the summation above.

\begin{lemma}\label{sum_nb}
If $X\overset{D}{=}NB(m_1,p)$ and $Y \overset{D}{=}NB(m_2,p)$, then $X+Y$ is $NB(m_1+m_2,p)$
\end{lemma}
\begin{proof}
Since $X$ is defined as the number of tails until $m_1$ heads when a coin with probability of heads, $p$ is repeatedly tossed and $Y$ represents the number of tails until $m_2$ heads, $X+Y$ will represent the number of tails until $m_1+m_2$ heads when this same coin is repeatedly tossed.
\end{proof}

\begin{conjecture}
If we collect $a$ observation periods from the control group and $b$ observation periods from the treatment group, the $\beta(\hat{\alpha}) \to 0$ as $a,b \to \infty$. If we collect only a single observation period from each group of length $t$, then, the $\beta(\hat{\alpha})$ converges asymptotically to a non-zero value as $t \to \infty$.
\end{conjecture}

Let's consider the special case when all the observation periods, $t_i$ across both groups are of the same length, $t$. So, we get a total of $a\times t$ observation period in the control group and $b\times t$ observation period in the treatment group. 

Now, the total number of events, $N_0$ in the control group is the sum of $a$ $NB(m,\frac{\theta}{\theta+t})$ random variables. So, by lemma \ref{sum_nb}, it is $NB(am, \frac{\theta}{\theta+t})$. Similarly, the number of events $N_1$ in the treatment group is $NB(b m, \frac{\theta_1}{\theta_1+t})$ where $\theta_1=\theta-\delta \theta$. By equation \ref{beta_rate_2}, we get the false negative rate:

Increasing the time periods for both groups is consistent with decreasing the values of $p_0$ and $p_1$ in equation \ref{eqn:false_neg_nbd}, while (under the assumption of equal interval lengths), increasing the number of observation time intervals is consistent with increasing the value(s) of $m$ (for $N_0$ and $N_1$). Doing the former for a type-1 error rate of $5\%$ produces figure \ref{beta_w_t_nbd} while doing the latter produces figure \ref{beta_m_nbd}.

We see that increasing the size of a single observation interval makes $\beta$ decrease, but it asymptotically approaches a finite value (about $34\%$ in this case) while increasing the size of the number of observation intervals makes $\beta$ decrease all the way to $0$.
\begin{figure}
  \includegraphics[width=0.8\linewidth]{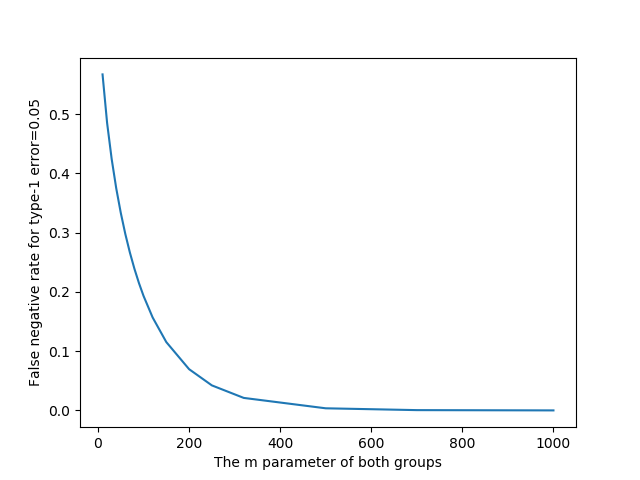}
  \caption{$\beta$ with $m$ of negative binomial}
  \label{beta_m_nbd}
\end{figure}

\begin{figure}
  \includegraphics[width=0.8\linewidth]{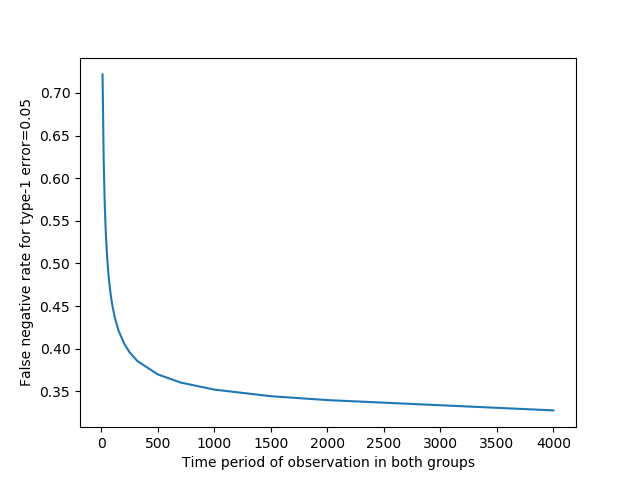}
  \caption{$\beta$ with $t$, size of observation period of both groups for negative binomial}
  \label{beta_w_t_nbd}
\end{figure}

\subsection{Comparison with the Wald test}
In \cite{zhu}, the authors consider hypothesis testing and sample size estimation for the negative binomial distribution for medical applications (making their work very relevant to this investigation). Although they focus more on estimating sample sizes and so, know some of the parameters in advance for that purpose, the sample size estimation has an implicit Wald test embedded therein which becomes an alternate approach to comparing failure rates. This subsection will be dedicated to a comparison of their test for comparing rates with the one we presented in equation \ref{rate_test}. Their formulation for the Negative Binomial point process is slightly different from but corresponds one to one with equation \ref{eqn:neg_binom}. We provide their formulation below as well as how it corresponds to ours. They define $Y_{ij}$ as the number of events during time $t_{ij}$ for subject $i$ in group $j$ ($j=0,1$, the control and treatment groups respectively). Then,

$$P(Y_{ij}=y_{ij}) = \frac{\Gamma(k^{-1}+y_{ij})}{\Gamma(k^{-1})y_{ij}!} \left(\frac{k \mu_{ij}}{1+k\mu_{ij}}\right)^{y_{ij}}  \left( \frac{1}{1+k\mu_{ij}} \right)^{\frac{1}{k}}$$

Here, $\mu_{ij}$ is the average number of events for subject $i$ in group $j$ and $\Gamma(.)$ is the Gamma function, a generalization of factorials. For any integer $m$ we have: $\Gamma(m)=(m-1)!$.

Matching this parameterization to equation \ref{eqn:neg_binom}, we also get:

$$\frac{1}{k} = m$$

$$\frac{t_{ij}}{t_{ij}+\theta} = \frac{k \mu_{ij}}{1+k\mu_{ij}}$$

Implying,

$$\frac{t_{ij}}{\theta} = k \mu_{ij}$$

So,

$$\mu_{ij} =\frac{m t_{ij}}{\theta} =E(L) t_{ij}$$

For devising the test, $\mu_{ij}$ is modeled as:

$$\log(\mu_{ij}) = \log(t_{ij})+\beta_0+\beta_1 x_{ij}$$

Here, $x_{ij}=I(j=1)$ is $1$ only if $j=1$. In other words, they assume a rate of 

\begin{equation}\label{zhu_ctrl_rate}\lambda = e^{\beta_0}\end{equation}

for the control group and 

\begin{equation}\label{zhu_trmt_rate}\lambda +\delta \lambda = e^{\beta_0+\beta_1}\end{equation} 

for the treatment (under the alternate hypothesis). 

$\beta_1$ then becomes the difference in the log-rates and the null hypothesis is predicated on it being normally distributed with mean $0$ (if there is no difference in rates, there is no difference in log-rates either).

The variance in $\beta_1$ is then estimated in equation (14) of their paper (changed some notation to avoid conflicts with this paper):

\begin{equation}\label{variance_def}V = \frac{1+\eta^2}{\mu_t (\lambda_0 + \eta \lambda_1)}+\frac{(1+\eta)k}{\eta}\end{equation}

Here, $\eta$ is the ratio of observations between treatment and control. If we use a balanced test (equal number of samples between the two groups), $\eta=1$, $\lambda_0$ and $\lambda_1$ are the rates for the two groups, $\mu_t$ is the expected number of events for each observation period and $k=\frac{1}{m}$, the inverse of the parameter of the negative binomial representing the number of heads desired. 

While $\mu_t$, $\lambda_0$ and $\lambda_1$ were easily estimated from the generated data, we weren't sure how to estimate the $k$ parameter. So, we simply plugged in the value used to generate the data in the first place. This is giving the test some information it shouldn't have (it should have access to only the data itself and nothing about the underlying process that generated said data). This is a potential flaw, but it can only help the Wald test, not hamper it.

With the variance, this hypothesis test simply becomes the standard Wald test where:

\begin{figure}
  \includegraphics[width=0.8\linewidth]{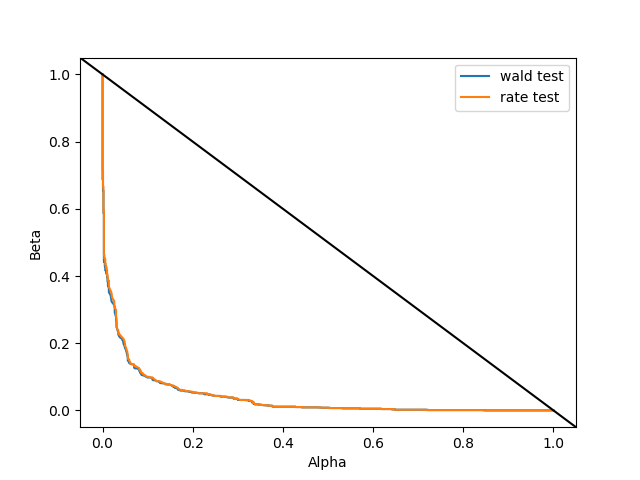}
  \caption{The $\alpha$-$\beta$ trade off for the Wald test and rate test for observation samples from the Negative Binomial distribution.}
  \label{wald_tst_comparison}
\end{figure}

\begin{itemize}
\item{Estimate the variance using equation \ref{variance_def} from the sample. Divide by sample size of control group and take square root to get standard deviation.}
\item{Find the difference in estimated log-rates from the data between treatment and control.}
\item{Find the inverse survival function of the normal distribution with mean $0$ and standard deviation calculated in step 1. This becomes the p-value of our test.}
\end{itemize}

We simulated some data from the null and alternate hypotheses as defined in table I of [3] and compared the performance of the two tests on an $\alpha$-$\beta$ curve. The result is shown in figure \ref{wald_tst_comparison}. It appears the two tests are completely on top of each other, with negligible difference in performance. We re-created this plot for various values of the free variables and the conclusion seemed to remain firmly the same. Not much to choose from between the two tests in terms of raw performance. 

This still motivates the use of the rate test over the Wald test for the following reasons:

\begin{itemize}
\item{The Wald test had to be customized for the Negative Binomial distribution (in terms of estimation of the variance) while the rate ratio test was used out of the box, ``as is'' and performed similarly. If we want to extend to some other point process apart from the negative Binomial tomorrow, we'd have some work to do for the Wald test (estimating variance) while the rate ratio test would be ready for application.}
\item{As mentioned previously, the Wald test had a bit of an unfair advantage in this experiment with regard to knowing the $k$ parameter used to generate the actual data, which probably helped its performance. The rate ratio test had no such advantage and still performed similarly.}
\item{The Wald test relies on the estimation of the variance, which blows up when we have $0$ events in one of the groups. The rate ratio test on the other hand, still produces sensible p-values.}
\item{The formula of the Wald test's p-value is more complex.}
\item{For the rate test, we can use equation \ref{fnr_fpr} to get the false negative rate in a form that is amenable to efficient numerical estimation, making things like sample size estimation much faster. For the Wald test, attempting to use the same expression results in complex expressions that can only be estimated with simulation.}
\end{itemize}

\section{Results and applications}
\subsection{Improving Azure customer experience with AIR}

Since recognizing that interruptions the rate at which they occur is an excellent proxy for measuring customer pain on the platform (and direct customer feedback played a big role in this) about a year ago, Azure has used ``Annual Interruption Rate'', defined as the number of virtual machine interruptions a typical Azure customer will experience if they ran 100-VM-years worth of workloads on the platform as a KPI. This is essentially a failure rate (calculated via equation \ref{rate_def}) with the ``100 VM-years'' chosen simply as a unit for the machine run-time in the denominator, designed to make the scale of the KPI look reasonable.

This has led to prioritization of fixes that drive this number down (which might have been ignored otherwise) and indeed, figure \ref{air_improvement} shows the long way Azure has come far in that regard, with the rate improving from about 70 a year ago to close to single digits today. This has of course contributed to greater customer satisfaction with the platform (particularly for customer workloads that are very sensitive to interruptions like massive online gaming servers), since the rate at which their workloads are interrupted has been trending in the right direction.

The effective technique for comparing AIR between two groups we have discussed in this paper has played a big role in this improvement, as we will see in subsequent subsections.

\subsection{Prioritizing AIR for small slices }
Like any KPI or statistic, we never know what the true interruption rate is but estimate it from a finite sample of data. These estimates carry some variance with them and the larger the time window (in terms of total observed VM-time), the less this variance becomes. When comparing large populations like two hardwares, we have enough of a sample size to simply work with the estimated AIR numbers themselves.

However, when slicing into much smaller buckets (like a single Azure machine or node), this variance becomes quite problematic. As mentioned in the previous sub-section, the rate of interruptions we see on Azure today is about 10 in 100-VM-years (or about 36500 VM-days).

Now, if we're observing a single node for a few (say 10) days running (typically) 10 VMs, this will make for 100 VM-days. So, the number of interruptions we expect to see is: $\frac{10 \times 100}{36500}\sim 0.03$. Of course, there is no such thing as fractional interruptions. In practice, this means we will see no interruptions at all 97 out of 100 times and about one interruption the other 3. This means that those 97 times, we'll observe an AIR of $0$ while the other $3$ times, an AIR of $\frac{1\times 36500}{100}=365$. This makes the metric quite fickle and we can imagine that no ranking we produce based on the raw number will carry any kind of stability.

It is here that the methodology for comparing failure rates as detailed in theorem \ref{rate_test} comes to the rescue. 
Instead of ranking on raw AIR estimates, we can take the number of reboots and observation duration (VM-time) for our small slice (treatment group) and compare with the same numbers for the past day applied to the general fleet (control group). We can then rank on the p-values we get, which tend to be much more stable ranking and help prioritize worthy issues by taking the noise as well as raw AIR estimates into account. 

The simplicity of the rate test has allowed us to code it up in a query language called Kusto (the data processing system of choice in Azure), which makes the barrier to usage across the organization very low and helps drive impact.

\begin{figure}
  \includegraphics[width=0.8\linewidth]{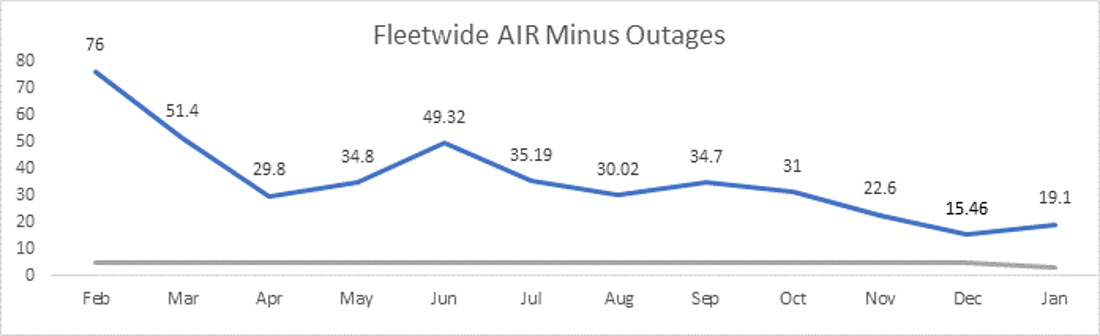}
  \caption{The improvement over time of ``Annual interruption rate'' for the Azure platform as a result of prioritization.}
  \label{air_improvement}
\end{figure}

\subsection{Statistical software testing for the cloud}
Microsoft Azure is an ever changing platform. The hardware running the machines that power it reach the end of their lives and get swapped out, the micro-code running on the chips that power the machines needs to be updated, the various agents running on those machines that help with the VM workflows need to be updated and so on. Most of these changes have the potential to cause regressions in the failure rate and for us to lose some of the ground we've covered over the past year. Since these bits are going to go to a complex cloud environment with diverse hardware configurations customer workloads, settings, etc. and with the customers running on this environment having a very high bar for platform availability and other metrics, the paradigm of traditional software testing needs to be extended along multiple fronts. The solution Azure is coming up with are pre-production testing environments that are being designed specifically to catch regressions that might occur when the payloads hit production. These environments run synthetic workloads, designed to mimic customer workloads on machines sampled from production and compare a control group without the new software change to a treatment group with it (A/B testing).

And since the core KPI we track is failure rate, the test described in theorem \ref{rate_test} holds a very important place in this effort. It is being used in some of these environments to test (for example) micro-code updates on the Intel chips that power Azure, updates to the host agent running on all Azure nodes, etc. where some issues have already been uncovered by it and are being actively investigated.

\subsection{Recommending time to wait}
Another application closely related to catching regressions with hypothesis testing is recommending the size of the test environment and the we should wait to collect data from our two groups before making a go-nogo decision on the new software bits.

Using the closed form expression for the false negative-false positive rate trade off detailed in section II and III, a simulator was created which anyone in Azure that want's to test a new feature can use for the following purposes:

\begin{itemize}
\item{Obtaining a mapping between $\alpha$ (desired FPR) and $\hat{\alpha}$ (type-1 error rate we should set), hence tuning the test to their data. This is per equation \ref{alpha_general}.}
\item{Plotting the profile of $\alpha$ vs $\beta$ to understand what trade off they will get on their data when using this test. This is per equations \ref{beta_rate_1}.}
\item{Given target effect size for failure rate regressions we want to be sensitive to, target false positive and false negative rates, how long should they wait for collecting data before making a go-nogo decision on their feature?}
\end{itemize}

For the time to wait application, using the rate test on $H_0$ and $H_a$ assuming the Poisson distribution provides a lower bound for the time to wait since the time-to-wait from a Poisson assumption will always be lower than that from real data. So, we should wait \textit{atleast} the amount of time a Poisson distributional assumption recommends for reaching certain target false positive and false negative rates for over-dispersed data.

\appendix

\section{Hazard rate}

Let's say that a process produces some events of interest (like motor accidents, machine failures, etc.). The time between successive occurrences of such events (inter-arrival time) is a random variable. Let's call it $T$. The probability density function (PDF) of this random variable is denoted by $f_T(t)$. By definition, the probability that we will see an event between some interval $(t, t+\delta t)$ is given by:

$$P(T \in (t,t+\delta t)) = f_T(t) \delta t$$

A quantity that is more useful in many contexts than the PDF is the hazard rate.

Let's say someone has had successful cancer treatment. The sad thing about cancer is that it is never completely cured and there is always a chance the body will relapse. Given this, a cancer patient might wonder: "it's been 1 year since my treatment and I haven't relapsed yet. Given I didn't relapse until now, what is the chance I'll relapse in the next month?". We can even remove the arbitrary 'month' interval in this statement and simply ask how many events do I expect to see per unit time. This becomes a 'rate' which is similar conceptually to velocity. Just as we have instantaneous velocity, we have instantaneous rate. We can express this notion mathematically as:

$$P(T \in (t,t+\delta t)\:| \: T > t) = \frac{P(T \in (t,t+\delta t) \,  \& \, T>t)}{P(T>t)}$$

$$= \frac{P(T \in (t,t+\delta t))}{P(T>t)}$$
By definition of the probability density function, $f_T(t)$ this becomes:

$$ = \frac{f_T(t) \delta t}{P(T>t)} $$

As $\delta t$ becomes small, the probability that more than one event will occur in that interval becomes negligible. So, there will be either $0$ or $1$ events in this interval when it is sufficiently small, effectively making the event a coin toss (a.k.a. a Bernoulli random variable). So, the probability calculated above is also the expected number of events in the small interval. Then, if the number of events per unit time is defined as the hazard rate function, $h_T(t)$; the number of events in a small interval proceeding $t$ will become: $h_T(t)\delta t$. Equating the two expressions we get:

$$h_T(t)\delta t = \frac{f_T(t)}{P(T>t)} \delta t$$

simplifying,

$$h_T(t) = \frac{f_T(t)}{P(T>t)} = \frac{f_T(t)}{S_T(t)}$$

Now, this rate is a function of time. Which means that in any given large enough interval of time, it will take on different values. What if we wanted to approximate it with a single number (say $\lambda$) for a given interval? By definition of averages, we would have:

$$\int_{t_1}^{t_2} h_T(t) dt = \int_{t_1}^{t_2} \lambda dt = \lambda(t_2-t_1)$$

Also, let's say that the number of events observed in the interval $(t_1,t_2)$ is given by the random variable, $N$. Then the expected value of $N$ is given by (by definition of $h_T$):

$$E(N) = \int_{t_1}^{t_2} h_T(t) dt = \lambda (t_2-t_1)$$

So we get:

\begin{equation}\lambda = \frac{E(N)}{t_2-t_1} \tag{2}\end{equation}

If we observe the process for a certain interval of time and count the number of events, $n$ within said interval then $n$ is an unbiased unbiased estimator for $E(N)$ and so the estimator for $\lambda$ becomes:

$$\hat{\lambda} = \frac{n}{t_2-t_1} = \frac{n}{\bigtriangleup t}$$

\section{Using the exponential distribution to estimate average rate}
Not only is the exponential distribution the simplest possible distribution for modeling the time until some event, it is also the only distribution that has a constant failure rate. This makes it a natural choice for estimating a single failure rate. Here, we will use it to obtain an unbiased estimator for the failure rate when some of our data is censored and some is un-censored. 

First, two quick properties of the exponential distribution. The probability density function is given by:

$$f_T(t) = \lambda e^{-\lambda t}$$

And the survival function is given by:

$$F_T(t) = \int\limits_t^\infty f_T(t) dt = e^{-\lambda t}$$

Using these, the likelihood function of data with $t_i$ un-censored data and $x_j$ censored data points becomes:

$$L(\lambda) = \prod\limits_{i=1}^n \lambda e^{-\lambda t_i} \prod\limits_{j=1}^m e^{-\lambda x_j}$$

Taking logs on both sides, we get the log likelihood function.

$$ll(\lambda) = \sum\limits_{i=1}^n (\log \lambda -\lambda t_i) - \sum\limits_{j=1}^m \lambda x_j$$

To get the value of the parameter $\lambda$ that minimizes the log-likelihood function, we take derivative with respect to $\lambda$ and set it to zero.

$$\frac{\partial ll(\lambda)}{\partial \lambda} = \frac{n}{\lambda} - \left(\sum\limits_i t_i + \sum\limits_j x_j\right)$$

$$\frac{\partial ll(\lambda)}{\partial \lambda} = 0$$

This gives us the MTTF (Mean time to failure):
$$\frac 1 \lambda = \frac{\sum t_i +\sum x_j}{n}$$

Hence, we can simply sum all the UP times (censored and uncensored) and divide by the number of downtime events to get the MTTF. A similar result holds for MTTR, replacing UP times with DOWN times. 

The AIR estimate then becomes: $\frac{1}{\text{MTTF}}$, which is just the total number of failures divided by the total UP time.

\section*{Acknowledgements}
This work exists solely because of people who were kind enough to discuss with us the applicability of the rate test to various kinds of point processes and the nuances therein. If these discussions had not happened, we would have simply applied the test, accepting that it is ``most powerful'' and not done all these experiments which led to a lot of new insight. Hence, we would like to thank Paul Li, Duncan Wardsworth, Yilan Zhang, Mary Hu and Oscar Zarate and Ze Li for being kind enough to take time to discuss and intelligently critique the methodology. 

In addition, Randolph Yao, Yifan Chang and Saurabh Agarwal provided insights that helped improve the paper.

We would also like to thank Naveen Gundavaram, Avinash Kumar and Arindam Basik as well as Meir Schmoley, Rahul Shah, Jayjit Phadke and Konstantine Maleshenko for their partnership in various test environment initiatives across Azure.


\end{document}